\documentclass{article}

\usepackage[utf8]{inputenc}
\date{}
\usepackage{microtype}
\usepackage{amsthm}
\usepackage{amssymb}
\usepackage{amsmath}
\usepackage{authblk}
\usepackage{graphicx}
\usepackage{tikz}
\usepackage{comment}

\newtheorem{theorem}{Theorem}
\newtheorem{proposition}{Proposition}

\newtheorem{definition}{Definition}

\newtheorem{lemma}[theorem]{Lemma}
\newtheorem{remark}{Remark}

\newtheorem{claim}{Claim}
\newtheorem{question}{Question}

\title{Decision DNNFs with imbalanced conjunction cannot efficiently represent CNFs of bounded width}
\author{Igor Razgon \\ Department of Computer Science, Durham University, igor.razgon@durham.ac.uk}
\begin{document}
\maketitle
\begin{abstract}
Decomposable Negation Normal Forms 
\textsc{dnnf} \cite{DarwicheJACM} is a landmark Knowledge Compilation (\textsc{kc})
model, highly important both in \textsc{ai} and Theoretical Computer Science. 
Numerous restrictions of the model have been studied. In this paper we consider the restriction
where all the gates are $\alpha$-imbalanced that is, at most one input of each 
gate depends on more than $n^{\alpha}$ variables (where $n$ is the number if variables of the function 
being represented). 

The concept of imbalanced gates has been first considered
in [Lai, Liu, Yin `New canonical representations by
augmenting OBDDs with conjunctive decomposition', JAIR, 2017].
We consider the idea in the context of representation of \textsc{cnf}s of bounded
primal treewidth. We pose an open question as to whether \textsc{cnf}s of bounded primal treewidth 
can be represented as \textsc{fpt}-sized \textsc{dnnf}
with $\alpha$-imbalanced gates. 

We answer the question negatively for Decision \textsc{dnnf} with 
$\alpha$-imbalanced conjunction gates. In particular, we establish a
lower bound of $n^{\Omega((1-\alpha) \cdot k)}$ for the representation 
size (where $k$ is the primal treewidth of the input \textsc{cnf}). 
The main engine for the above lower bound is a combinatorial
result that may be of an independent interest in the area of parameterized 
complexity as it introduces a novel concept of bidimensionality.
\end{abstract}
Non-local width parameters of graphs\\
Lower bounds\\
Fixed-parameter tractability
\section{Introduction}
\subsection{Statement of the result and motivation}
Decomposable Negation Normal Forms 
\textsc{dnnf} \cite{DarwicheJACM} is a landmark Knowledge Compilation (\textsc{kc})
model, highly important both in \textsc{ai} and Theoretical Computer Science 
\cite{Korhonen21}. The model can be defined as a De Morgan circuit $B$ where the conjunction gates 
are \emph{decomposable}. To define the decomposability, let us agree 
that for a node $t \in V(B)$, $Var(t)$ is the set of variables $x$ such that
$B$ has a path to $t$ from a node labelled by $x$ or by $\neg x$. 
Suppose that $t$ is a conjunction gate with input gates $t_1$ and $t_2$. 
Then the decomposability means that $Var(t_1) \cap Var(t_2)=\emptyset$. 
In this paper, inspired by \cite{ANDOBDDuneven}, we introduce a further restriction 
on the conjunction gates making them \emph{imbalanced}. In particular, we fix 
an $0< \alpha<1$ and require that for each conjunction gate $t$,  there is 
at most one input $t'$ such that $Var(t')>n^{\alpha}$ where $n$ is the number of variables
of the whole \textsc{dnnf}. We call such a gate $\alpha$-\emph{imbalanced}. 
Thus if $t$ has many input variables then the 
vast majority of them are input variables of \emph{exactly one} input gate of $t$. 
\footnote{In the original definition in \cite{ANDOBDDuneven}, the definition of an imbalanced
gate does not depend on $n$. In particular, it is required that at most one 
input of $t$ has more than $r$ input variables where $r$ is a fixed number.}

It is well known \cite{DarwicheJACM} 
that \textsc{dnnf} admit an \textsc{fpt}-sized representation
for \textsc{cnf}s of bounded primal treewidth. We are interested
if the same is true if the conjunction gates are imbalanced. 
\begin{question} \label{quest:maindnnf}
Does a \textsc{cnf} $\varphi$ have a representation as a \textsc{dnnf} with 
all the conjunction gates being $\alpha$-imbalanced that is 
of size \textsc{fpt} in the number of variables  of $\varphi$ parameterized
by the primal treewidth of $\varphi$? 
\end{question} 

Question \ref{quest:maindnnf} is well motivated due to two reasons. 
First, to the best of our knowledge, all the known \textsc{dnnf} restrictions
e.g. \cite{SDD,DesDNNF} do retain the power of \textsc{fpt}-sized representation of 
\textsc{cnf}s of bounded primal treewidht. Another motivation arises if we view
\textsc{dnnf} as a non-deterministic read-once branching program ($1$-\textsc{nbp})
equipped with decomposable conjunction gates. 
In this context, the \textsc{dnnf} with imbalanced conjunction gates lies
in-between the $1$-\textsc{nbp}, where no conjunction gates are used, and
\textsc{dnnf} where fully fledged decomposable conjunction gates are used.
As the former is known to be non-\textsc{fpt} on \textsc{cnf}s of bounded treewidth \cite{RazgonAlgo}
and the latter is \textsc{fpt} \cite{DarwicheJACM} , it is interesting to find out what happens in-between. 

In this paper, we leave Question \ref{quest:maindnnf} open and instead 
resolve its restricted version for Decision \textsc{dnnf} \cite{DesDNNF} rather than 
the \textsc{dnnf} as a whole. 
\begin{theorem} \label{th:mainintro}
For each sufficiently large $k$, there is an infinite  class ${\bf \Phi}_k$
of \textsc{cnf}s of primal treewidth at most $k$ whose representation as a Decision
\textsc{dnnf} with $\alpha$-imbalanced gates has size $n^{\Omega((1-\alpha) \cdot \sqrt{k})}$. 
\end{theorem}
Theorem \ref{th:mainintro} can be seen as a deterministic version of Question \ref{quest:maindnnf}.
Indeed,  Decision \textsc{dnnf} can be seen as a \emph{deterministic} read-once branching program
$1$-\textsc{bp} equipped with decomposable conjunction gates \cite{BeameDNNF} while \textsc{dnnf} is the analogous upgrade
for $1$-\textsc{nbp}. 

\subsection{Overview of our approach}
Let us start from defining the set of hard instances that we use for proving 
Theorem \ref{th:mainintro}. 
For a graph $G$, let $\varphi(G)$ be a \textsc{cnf} over $V(G)$ with the clauses
$\{(u \vee v)|\{u,v\} \in E(G)\}$. 
For each sufficiently large $k$, we establish the lower bound of Theorem \ref{th:mainintro} 
for a class of \textsc{cnf}s $\varphi(G)$ 
where $G$ is a family of graphs of treewidth at most $k$ as defined below. 

Let $T[h]$ be a complete ternary tree of height $h$. We regard $T[h]$ as an undirected graph
with a specified root and hence the descendant relation. Recall that a path of $k$ vertices is denoted by $P_k$. 
In order to prove the lower bound we use the graphs $T[h] \square P_k$ for a sufficiently large $k$. 
The sign $\square$ denotes the Cartesian product of two graphs. In particular, $H_1 \square H_2$ is a graph on $V(H_1) \times V(H_2)$
such that $(u_1,u_2)$ is adjacent to $(v_1,v_2)$ is and only if $u_1=v_1$ and $\{u_2,v_2\} \in E(H_2)$ or if
$\{u_1,v_1\} \in E(H_1)$ and $u_2=v_2$. For the sake of brevity, we refer to $T[h] \times P_k$ as $T[h,k]$.
The graphs $T[1]$ and $T[1.2]$ are illustrated on Figure \ref{fig:treeproduct}. 

In order to prove Theorem \ref{th:mainintro}, we view Decision \textsc{dnnf} as a $1$-\textsc{bp}
equipped with decomposable conjunction gates. As $1$-\textsc{bp} is well known under the name
Free Binary Decision Diagaram (\textsc{fbdd}), throughout the paper, we refer to
Decision \textsc{dnnf} as $\wedge_d$-\textsc{fbdd} and as $\wedge_{d,\alpha}$-\textsc{fbdd} 
to the special case where all the conjunction gates are $\alpha$-imbalanced.

Our approach to proving lower bounds for $\wedge_{d,\alpha}$-\textsc{fbdd} 
is an upgrade of an approach for proving lower bounds for \textsc{fbdd} \cite{RazgonAlgo}, therefore
we consider the latter first. 
The main engine for proving an \textsc{fbdd} lower bound for $\varphi(G)$ such that
$G$ is of a constant max-degree is demonstrating that $G$ has a large 
linear maximal induced matching width (linear \textsc{mim} width) \cite{VaThesis,obddtcompsys}. 
In particular, linear \textsc{mim}-wdith at least $p$  means that any permutation
$\pi$ of $V(H)$ (viewed as  a linear order of $V(G)$) can be partitioned into a prefix $\pi_0$
and the corresponding suffix $\pi_1$ so that there is an induced matching $M$ of size at least $p$
one end of each edge $e \in M$ belongs to $\pi_0$, the other to $\pi_1$. 
This situation is schematically illustrated on the left-hand part of Figure \ref{fig:permstruct}, 
the vertices in $\pi_0$ coloured in red, the vertices in $\pi_1$ 
are coloured in green. 
It can be shown that the linear \textsc{mim} width of $T[h,k]$ is $\Omega(h \cdot k)$ or,
to put it differently, $\Omega(\log n \cdot k)$ where $n=|V(T[h,k])|$. 

In order to demonstrate how the above $\Omega(\log n \cdot k)$ bound yields an \textsc{xp}
lower bound for \textsc{fbdd} representing $T[h,k]$, we need to extend our terminology. 
Let $B$ be an \textsc{fbdd} and let $P$ be a directed path of $B$. 
Let $Var(P)$ be the set of variables labelling the nodes of $P$ but the last one. 
Let ${\bf a}(P)$ be the assignment of $P$ where each $x \in Var(P)$ is assigned by the label
of the edge of $P$ leaving the node labelled by $x$. 

Let us say that $P$ is a \emph{target path} if it starts at the source of $B$.
We say that an assignment ${\bf g}$ is \emph{carried} by $P$ if ${\bf a}(P) \subseteq {\bf g}$. 
We also say in this case  that ${\bf g}$ is \emph{carried through} $u(P)$ where $u(P)$ is the 
final node of $P$. 

Let ${\bf g}$ be a satisfying assignment of $\varphi(T[h,k])$.
By definition of \textsc{fbdd}, there is a target path $P$ ending at the $True$ sink of $P$
carrying ${\bf g}$. Using the $\Omega(\log n \cdot k)$ lower bound, 
we demonstrate that $P$ can be partitioned into a prefix $P_0$ and the corresponding suffix $P_1$
so that that there is an induced matching $M$ of $T[h,k]$ such that one end of each edge of the matching 
belongs to to $Var(P_0)$ and the other end to $Var(P_1)$. 
We refer  to $u(P_0)$ as $u({\bf g})$. 
We then define a probability space over the set of satisfying assignments of $\varphi(T[h,k])$ 
and demonstrate that the probability of carrying such an assignment through $u({\bf g})$
is $O(n^{-k})$. Hence, by the union bound the number of distinct $u({\bf g})$ nodes must
$\Omega(n^k)$. 

The picture of the left-hand part of Figure \ref{fig:permstruct} can illustrate the above reasoning
with the vertices of $V(M) \cap Var(P_0)$ coloured in red and the vertices in $V(M) \cap Var(P_1)$
coloured in green. We demonstrate that, per each edge connecting a red and a green vertex, we can 
pick one vertex and the set of all picked vertices are assigned positively for all 
the assignments carried through $u({\bf g})$. In other words, the assignments of $\Omega(\log n \cdot k)$
variables is fixed and this implies the probability $O(n^{-k})$ bound. 

The above argument does not work for $\wedge_d$-\textsc{fbdd}. 
The reason for that 
is that decomposable $\wedge$ nodes can split the  variables apart. 
In particular,it is possible to construct a $\wedge_d$-\textsc{fbdd}
representing $\varphi(T[h,k])$ so that every matching between a prefix
and the corresponding suffix of a source-sink path 
is of size $O(k)$ (and such a $\wedge_d$-\textsc{fbdd} will have an \textsc{fpt}
size parameterized by $k$). It  turns out, however that imbalanced gates
are incapable to implement this efficient splitting. To demonstrate this,
we need an upgraded version of the combinatorial statement as demonstrated 
on the right-hand part  of Figure \ref{fig:permstruct}.

In the upgraded version, an additional constraint is imposed on the suffix of the permutation. 
Each green vertex now needs to belong to a connected component of the  subgraph of $T[h,k]$
induced  by the suffix that has size larger than $n^{\alpha}$. 
This connected component serves as an `anchor' that does not let the green vertices
to be  split apart: the conjunction gate splitting any two such components cannot be 
$\alpha$-imbalanced. Then an upgraded version of the argument for \textsc{fbdd} can be
applied. 

The upgraded argument takes into account that in $\wedge_d$-\textsc{fbdd}
(even  with imbalanced gates), assignments are no longer associated with paths
but rather with trees whose internal nodes correspond to the conjunction nodes
of the $\wedge_d$-\textsc{fbdd}. Since the conjunction gates are imbalanced, we identify a 
\emph{mainstream} path by going from the source to the sink and choosing the 
`heavy' child for each cojunction gate (unless both children are light and then an arbitrary one 
can be chosen). 
Then effectively the same argument as for \textsc{fbdd} (with several more rather natural
tweaks) can be applied to the mainstream paths. 

In  the context of width parameters, the right-hand part of 
Figure \ref{fig:permstruct} illustrates a `non-local' width measure. 
The `local' part of this measure is the induced matching size 
between a prefix and the corresponding suffix of the considered 
permutation. The `non-local' part is the connectivity condition imposed
on the subgraph induced by the vertices of the suffix. 
This width measure might be of an independent interest as it may lead to a new concept of
bidimensionality \cite{bidimensionality}, we provide a more 
detailed discussion in the Conclusion section. 
The idea of non-local condition is conceptually related 
to hypertree width \cite{HTWpaper} which is a generalized hypertree width plus
a non-local condition called \emph{special condition}. This special condition
guarantees \textsc{xp}-time computability of the hypertree width \cite{HTWpaper}
as opposed to the \textsc{np}-hardness of the generalized version  \cite{2009gottlob}. 

\subsection{Structure of the paper}
In Section \ref{sec:prelim}, we introduce the necessary background.
The main result is stated and proved in Section \ref{sec:main}.
The proof uses several auxiliary statements whose proofs are, in turn, provided 
in Subsections \ref{sec:harness1}, \ref{sec:harness2}, and Section \ref{sec:maincomb}.
In particular, Section \ref{sec:maincomb} provides a proof for the main comnbinatorial engine
as illustrated on the right-hand part of Figure \ref{fig:permstruct}, 
and Subsections \ref{sec:harness1} and \ref{sec:harness2} provide proofs of statements 
needed to harness the engine into the proof of the main lower bound statement. 
Section \ref{sec:conclusion} discusses future research. 
The paper also includes an appendix where we prove a technical lemma stated 
in the Preliminaries section.

\begin{figure}[h]
\begin{tikzpicture}

\draw [fill=black]  (1,1) circle [radius=0.2];
\draw [fill=black]  (2,1) circle [radius=0.2];
\draw [fill=black]  (3,1) circle [radius=0.2];
\draw [fill=black]  (2,2) circle [radius=0.2];
\draw [ultra thick] (1,1) --(2,2); 
\draw [ultra thick] (2,1) --(2,2); 
\draw [ultra thick] (3,1) --(2,2); 

\draw [fill=black]  (5,1) circle [radius=0.2];
\draw [fill=black]  (5.5,1) circle [radius=0.2];
\draw [fill=black]  (6.5,1) circle [radius=0.2];
\draw [fill=black]  (7,1) circle [radius=0.2];
\draw [fill=black]  (8,1) circle [radius=0.2];
\draw [fill=black]  (8.5,1) circle [radius=0.2];
\draw [ultra thick] (5,1) --(5.5,1);
\draw [ultra thick] (6.5,1) --(7,1);
\draw [ultra thick] (8,1) --(8.5,1);
\draw [fill=black]  (6.5,2) circle [radius=0.2];
\draw [fill=black]  (7,2) circle [radius=0.2];
\draw [ultra thick] (6.5,2) --(7,2);
\draw [ultra thick] (5,1) --(6.5,2);
\draw [ultra thick] (5.5,1) --(7,2);
\draw [ultra thick] (6.5,1) --(6.5,2);
\draw [ultra thick] (7,1) --(7,2);
\draw [ultra thick] (8,1) --(6.5,2);
\draw [ultra thick] (8.5,1) --(7,2);
\end{tikzpicture}
\caption{$T[1]$ and $T[1,2]$}
\label{fig:treeproduct}
\end{figure}
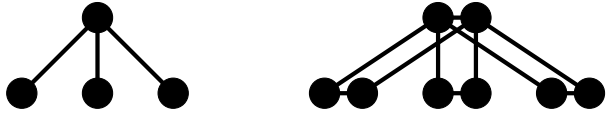

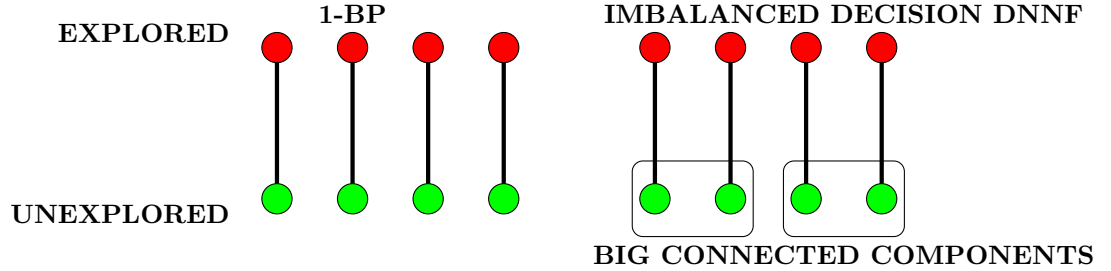
\begin{figure}[h]
\begin{tikzpicture}
\draw [fill=red]  (1,4) circle [radius=0.2];
\draw [fill=red]  (2,4) circle [radius=0.2];
\draw [fill=red]  (3,4) circle [radius=0.2];
\draw [fill=red]  (4,4) circle [radius=0.2];

\draw [fill=green]  (1,2) circle [radius=0.2];
\draw [fill=green]  (2,2) circle [radius=0.2];
\draw [fill=green]  (3,2) circle [radius=0.2];
\draw [fill=green]  (4,2) circle [radius=0.2];

\draw [fill=red]  (6,4) circle [radius=0.2];
\draw [fill=red]  (7,4) circle [radius=0.2];
\draw [fill=red]  (8,4) circle [radius=0.2];
\draw [fill=red]  (9,4) circle [radius=0.2];

\draw [fill=green]  (6,2) circle [radius=0.2];
\draw [fill=green]  (7,2) circle [radius=0.2];
\draw [fill=green]  (8,2) circle [radius=0.2];
\draw [fill=green]  (9,2) circle [radius=0.2];

\draw [ultra thick] (1,3.8) --(1,2.2); 
\draw [ultra thick] (2,3.8) --(2,2.2); 
\draw [ultra thick] (3,3.8) --(3,2.2); 
\draw [ultra thick] (4,3.8) --(4,2.2); 

\draw [ultra thick] (6,3.8) --(6,2.2); 
\draw [ultra thick] (7,3.8) --(7,2.2); 
\draw [ultra thick] (8,3.8) --(8,2.2); 
\draw [ultra thick] (9,3.8) --(9,2.2); 

\draw [rounded corners] (5.7,2.5)  rectangle (7.3,1.5);
\draw [rounded corners] (7.7,2.5)  rectangle (9.3,1.5);

\node [left]  at (0.5,4.2)  {\bf EXPLORED};
\node [left]  at (0.5,1.8)  {\bf UNEXPLORED};
\node [above]  at (2,4.2)  {\bf 1-BP};
\node [above]  at (8.5,4.2)  {\bf IMBALANCED DECISION DNNF};
\node [below]  at (8.5,1.5)  {\bf BIG CONNECTED COMPONENTS};
\end{tikzpicture}
\caption{Illustration of structural properties of permutations of $V(G)$  for branching programs 
(left) and for imbalanced Decision DNNFs (right)}
\label{fig:permstruct}
\end{figure}

\section{Preliminaries} \label{sec:prelim}
\subsection{Assignments, their sets, and rectangles}
Let $X$ be a finite set of variables.
We call a function ${\bf a}: X \rightarrow \{0,1\}$ an 
assignment (over $X$). We also refer to $X$ as $Var({\bf a})$. 
(In the rest of the paper, we often use the notation $Var(\mathcal{O})$
 denoting the set of variables of the object $\mathcal{O}$.)
The \emph{projection} of ${\bf a}$ to a set $Y$ of variables,
denoted by $Proj({\bf a},Y)$ is an assignment over $X \cap Y$ that
maps each $x \in X \cap Y$ to ${\bf a}(x)$. 
Let ${\bf b}$ be an assignment such that $Proj({\bf a},Var({\bf a}) \cap Var({\bf b}))=
Proj({\bf b}, Var({\bf a}) \cap Var({\bf b}))$. 
Then the assignments ${\bf a} \cup {\bf b}$, ${\bf a} \cap {\bf b}$,
${\bf a} \setminus {\bf b}$ and ${\bf b} \setminus {\bf a}$ are naturally
defined as the  corresponding operations on functions. 

Throughout this paper, when we consider a set $\mathcal{H}$ of assignments,
they are over the same set of variables. Therefore, if $\mathcal{H} \neq \emptyset$,
we define $Var(\mathcal{H})$ as $Var({\bf g})$ for an arbitrary ${\bf g} \in \mathcal{H}$
(if $\mathcal{H}=\emptyset$, let us set $Var(\mathcal{H})=\emptyset$). 
Let $Y$ be a set of variables. The \emph{projection} of $\mathcal{H}$ to $Y$,
denoted by $Proj(\mathcal{H},Y)$ is $\{Proj({\bf g},Y)| {\bf g} \in \mathcal{H}\}$. 
Let ${\bf a}$ be an assignment and 
let $\mathcal{H}|_{\bf a}=\{{\bf b} \setminus {\bf a}| {\bf b} \in \mathcal{H}, 
Proj({\bf a},Var(\mathcal{H})) \subseteq {\bf b}\}$. 
We refer to $\mathcal{H}|_{\bf a}$ as the \emph{restriction} of $\mathcal{H}$
by ${\bf a}$. 

For example, if 
$\mathcal{H}=\{
\{(x_1,0),(x_2,0),(x_3,1)\},
\{(x_1,1),(x_2,1),(x_3,0)\}
\}$ and ${\bf a}=\{(x_1,1)\}$
then $\mathcal{H}|_{\bf a}=\{\{(x_2,1),(x_3,0)\}\}$. 
Note that we could have used ${\bf a}=\{(x_1,1),(x_4,0)\}$
with exactly the same effect. 

\begin{remark} \label{rem:nonsubset}
Note that, in the definitions of a projection, we allowed
that the set of variables an assignment or a set of assignments is projected 
on is not necessarily a subset of the set of variables of that assignment or a set of assignments. 
Likewise, in the definition of a restriction, we allowed that the set
of variables of the restricting assignment is not necessarily a subset of the
set of variables of the set of assignments being restricted. 
This extra generality will help us to streamline our reasoning about
$\wedge_d$-\textsc{fbdd}s in that we will not need to consider unnecessary lateral
cases. 
\end{remark}

We let $\{0,1\}^X$ be the set of all assignments over a set $X$. 
We refer to $f: \{0,1\}^X \rightarrow \{0,1\}$ as a \emph{Boolean
function} over $X$ (that is, $Var(f)=X$). 
An assignment ${\bf g} \in \{0,1\}^X$ is a \emph{satisfying assignment} of $f$
if $f({\bf g})=1$. 
We denote by $\mathcal{S}(f)$ the set of all the satisfying assignments of $f$. 
Note that a Boolean function over the given set of variables is uniquely determined by its
set of satisfying assignments. 
Let ${\bf a}$ be an assignment. We denote by $f|_{\bf a}$ the Boolean function
over $Var(f) \setminus Var({\bf a})$ such  that
$\mathcal{S}(f_{\bf a})=\mathcal{S}(f)|_{\bf a}$. 

Let $\mathcal{H}_1$ and $\mathcal{H}_2$ 
be sets of assignments such that $Var(\mathcal{H}_1) \cap Var(\mathcal{H}_2)=\emptyset$. 
We define $\mathcal{H}_1 \times \mathcal{H}_2=\{{\bf a} \cup {\bf b}| 
{\bf a} \in \mathcal{H}_1, {\bf b} \in \mathcal{H}_2\}$. Note a slight abuse 
of notation in the use of the Cartesian product symbol for what is, strinctly 
speaking, not a Cartesian product. 
We call the set $\mathcal{H}$ such that there are $\mathcal{H}_1$ and $\mathcal{H}_2$
as above such that $\mathcal{H}=\mathcal{H}_1 \times \mathcal{H}_2$ a \emph{rectangle}. 

Note that if $\mathcal{H}_2=\emptyset$ then $\mathcal{H}_1 \times \mathcal{H}_2=\emptyset$.
However, $\mathcal{H}_2=\{\emptyset\}$ then $\mathcal{H}_1 \times \mathcal{H}_2=\mathcal{H}_1$
Thus, every set of assignments can be viewed as a rectangle. We are interested however 
in \emph{non-trivial rectangles} $\mathcal{H}_1 \times \mathcal{H}_2$ where 
$Var(\mathcal{H}_i) \neq \emptyset$ for each $i \in \{1,2\}$. In this paper, we use 
rectangles in the context of them breaking or not breaking sets of variables, the corresponding
notion is defined below. 

\begin{definition} \label{def:varbreak}
Let $\mathcal{H}$ be a set of assignments and let $Y$ be a set of 
variables. We say that $\mathcal{H}$ \emph{breaks} $Y$ if 
$\mathcal{H}=\mathcal{H}_1 \times \mathcal{H}_2$ so that both
$Var(\mathcal{H}_1) \cap Y$ and $Var(\mathcal{H}_2) \cap Y$ are non-empty. 
\end{definition} 
For example if $\mathcal{H}=\mathcal{H}_1 \times \mathcal{H}_2$ 
and $Var(\mathcal{H}_1)=\{x_1,x_2\}$ and
$Var(\mathcal{H}_2)=\{x_3,x_4\}$ then $\mathcal{H}$ breaks $\{x_2.x_3\}$
but does not break $\{x_1,x_2\}$. 

For sets $\mathcal{H}_1, \dots \mathcal{H}_q$ of assignments over pairwise 
disjoint sets of variables and $q>2$, the operation
$\prod_{i=1}^q \mathcal{H}_i$ is defined as $(\prod_{i=1}^{q-1}) \times \mathcal{H}_q$.
Clearly, the order of the sets in the product does not matter. 
\begin{proposition}
Let $\mathcal{H}=\prod_{i=1}^q \mathcal{H}_i$ for $q \geq 2$. 
Let $Y \subseteq Var(\mathcal{H})$ such that $\mathcal{H}$
does not break $Y$. Then there is $i \in [q]$ such that
$Y \subseteq Var(\mathcal{H}_i)$.
\end{proposition}

\begin{proof}
Indeed, assume that such an $i$ does not exist. 
Then we can identify $i_1 \neq i_2 \in [q]$ such that
both $Var(\mathcal{H}_{i_1}) \cap Y$ and 
$Var(\mathcal{H}_{i_2}) \cap Y$ are non-empty. 
Let $I_1,I_2$ be an arbitrary partition of $[q]$ such that
$i_1 \in I_1$ and $i_2 \in I_2$. 
Then $\mathcal{H}=\prod_{i \in I_1} \mathcal{H}_i \times \prod_{i \in I_2} \mathcal{H}_2$
breaks $Y$, a contradiction.
\end{proof}

\subsection{Graphs, treewidth, Cartesian products, permutations}
We use the standard terminology for graphs as in e.g. \cite{Diestel3}.
In particular, for a graph $G$ and $U \subseteq V(G)$, $G[U]$ denotes the subgraph
of $G$ \emph{induced} by $U$. We say that $U$ is a \emph{connected set} meaning that $G[U]$
is connected. $U$ being an \emph{independent set} means that $G[U]$ has no edges. 

\begin{definition}
Let $H_1$ and $H_2$ be two graphs. 
The \emph{Cartesian product} $H_1 \square H_2$ of
$H_1$ and $H_2$ is a graph with the set of vertices being
$V(H_1) \times V(H_2)$ and the set of edges
$\{\{(u,v_1),(u,v_2)\}|u \in V(H_1), \{v_1,v_2\} \in E(H_2)\} \cup
\{\{(u_1,v),(u_2,v)\}|\{u_1,u_2\} \in E(H_1), v \in V(H_2) \}$.
\end{definition}

\begin{definition}
We denote by $T[h]$ the complete rooted ternary tree of height $h$
as recursively defined below. 
$T[0]$ has a singleton set of vertices and no edges,
$root(T[0])$ is the only vertex of $T[0]$. 
Suppose that $h>0$ and let $T_1,T_2,T_3$ be three disjoint copies 
of $T[h-1]$. Then $T[h]$ is obtained by introducing a new vertex
$u$ and making it adjacent to $root(T_1),root(T_2),root(T_3)$.
We set $root(T[h])=u$. 

We denote the leaves and the internal vertices of $T[h]$ 
by $Leaves(T[h])$ and $Inter(T[h])$ respectively

The introduction of a root naturally introduces the descendancy 
relation. In particular, for $t \in Inter(T[h])$, 
$Children(t)$ denotes the set of children of $t$. 
\end{definition}

We denote by $P_k$ a path of $k$ vertices.
The graph $T[h] \square P_k$ plays an important role
for our reasoning. For the sake of brevity we denote
the graph as $T[h,k]$. 

We do not define the notions of tree decomposition and treewidth
since we only need them to establish the following (easy to prove)
statement. 

\begin{proposition} \label{prop:twup}
The treewidth of $T[h,k]$ is at most $2k-1$.
\end{proposition}

Next, we define an important notation that
we use for directed acylic graphs (\textsc{dag}s) 
and rooted trees. 
\begin{definition}
Let $D$ be a \textsc{dag} and let $u \in V(D)$.
Then $D_u$ is the subgraph of $D$ induced by all the vertices reachable
from $u$ (obviously, including $u$ itself), the labels on vertices
and edges (if any) are preserved as in $D$. 
Similarly, if $T$ is a rooted tree and $t \in V(T)$
then $T_t$ is the subgraph of $T$ induced by all the descendants of $t$
(including $t$ itself). Clearly, $root(T_t)=t$. 
We also remark that we sometimes refer to vertices of a \textsc{dag} 
as \emph{nodes}. 
\end{definition}

For our reasoning we extensively use permutations 
regarding them as linear orders. For example,
$\pi=(1,4,2,3,5)$ is a permutation of $[5]$. 
Such a setting enables us to refer to suffixes and prefixes 
of of permutations. For example 
$(1,4,2)$ is a prefix of $\pi$ and 
$(3,5)$ is the corresponding suffix. 
We naturally define a permutation \emph{induced}
by a subset of its elements using the 
same notation as for induced subgraphs. For example, $\pi[\{3,4,5\}]=(4,3,5)$. 
We may also identify linear orders with their underlying sets 
writing e.g. $4 \in (1,4,2)$. The correct use will always be clear from 
the context. 

\subsection{CNFs representing graphs}
We assume that the reader is familiar with the notion
of Conjunctive Normal Forms (\textsc{cnf}s). 
We identify a \textsc{cnf} $\varphi$ with the set of its 
clauses. In this paper we will consider only \textsc{cnf}s
representing graphs as defined below. 

\begin{definition}
Let $G$ be a graph without isolated vertices. 
Then $\varphi(G)$ is the \textsc{cnf} with the 
set of variables $Var(\varphi(G))$ being $V(G)$
and the clauses $\{(u \vee v)| \{u,v\} \in E(G)\}$. 
\end{definition}

We will identify $\varphi=\varphi(G)$ with the function
it represents so we will use $\mathcal{S}(\varphi)$ to denote
the set of satisfying assignments of $\varphi$. 
Let ${\bf g}$ be an assignment over a subset of $V(G)$. 
Then ${\bf g}$ \emph{falsifies} $(u \vee v)$ if 
$u,v \in Var({\bf g})$ and ${\bf g}(u)={\bf g}(v)=0$. 
It is not hard to see that ${\bf g}$ is a satisfying 
assignment of $\varphi$ if and only if $Var({\bf g})=V(G)$
and ${\bf g}$ does not falsify any clause of $\varphi$.
We say that ${\bf g}$ \emph{fixes} $v \in V(G)$ 
if there is $u \in Var(\varphi)$ such that ${\bf g}(u)=0$
and $(u \vee v)$ is a clause of $\varphi$. Further on, ${\bf g}$
fixes $U \subseteq V(G)$ if it fixes at least one element of $U$ or, to put
it differently, $U$ is unfixed by ${\bf g}$ if no element of $U$
is fixed by ${\bf g}$. 
Clearly any satisfying assignment of $\varphi$ that is a superset
of ${\bf g}$ must assign positively all the variables fixed by ${\bf g}$.

The following two statements are important for our further reasoning.  

\begin{lemma} \label{lem:nobreak0}
Let ${\bf g}$ be an assignment to a subset of $V(G)$ 
that does not falsify any clauses. 
Let $u \in V(G) \setminus Var({\bf g})$ be a variable 
unfixed by ${\bf g}$. 
Let ${\bf g}_u$ be an extension of ${\bf g}$
that assigns $u$ negatively and assigns positively 
the rest of the variables of $Var(G) \setminus (Var({\bf g}) \cup \{u\})$
Then ${\bf g}_u$ is a satisfying assignment of $\varphi(G)$. 
\end{lemma}

\begin{proof}
Let $(u_1 \vee u_2)$ be a clause of $\varphi(G)$. 
By the assumption about ${\bf g}$, $(u_1 \vee u_2)$ is satisfied 
if $\{u_1,u_2\} \subseteq Var({\bf g})$. 
If $\{u_1,u_2\} \subseteq V(G) \setminus Var({\bf g})$
then only one of $u_1,u_2$ can be assigned negatively (if it equals $u$)
hence the clause is again satisfied. 
Finally, assume that, say $u_1 \in Var({\bf g})$ whereas 
$u_2 \notin Var({\bf g})$. If $u_2=u$ then $u_1$ must be assigned positively
due to $u$ being not fixed. Otherwise, $u_2$ is assigned positively. 
In both cases $(u_1 \vee u_2)$ is satisfied. 
\end{proof}

\begin{lemma} \label{lem:nobreak1}
Assume ${\bf g}$ does not falsify any clause 
of $\varphi(G)$ and that $U$ is a connected set of $G$ unfixed by ${\bf g}$. 
Then $\mathcal{S}(\varphi)|_{\bf g}$ does not break $U$. 
\end{lemma}

\begin{proof}
Assume that the statement of the lemma is not true. 
Then $\mathcal{S}(\varphi(G))|_{\bf g}=\mathcal{S}_1 \times \mathcal{S}_2$
is a rectangle such that for each $i \in [2]$, $U_i=U \cap Var(\mathcal{S}_i) \neq \emptyset$. 
Since $U$ is connected, there is an edge $\{u_1,u_2\}$ of $G$ such that $u_1 \in U_1$
and $u_2 \in U_2$ 
Let ${\bf g^*}_i=Proj({\bf g}_{u_i} \setminus {\bf g},U_i)$ for each $i \in [2]$. 
It follows from Lemma \ref{lem:nobreak0} that ${\bf g^*}_1 \in \mathcal{S}_1$
and ${\bf g^*}_2 \in \mathcal{S}_2$. 
Therefore, ${\bf g^*}_1 \cup {\bf g^*}_2 \in \mathcal{S}_1 \times \mathcal{S}_2$. 
Hence, ${\bf g} \cup {\bf g^*}_1 \cup {\bf g^*}_2$ is a satisfying assignment of $\varphi(G)$. 
However, this is a contradiction since ${\bf g} \cup {\bf g^*}_1 \cup {\bf g^*}_2$
falsifies $(u_1 \vee u_2)$.
\end{proof}

\subsection{Decision DNNFs (a.k.a.) $\wedge_d$-\textsc{fbdd}s}
It is well known that Decision \textsc{dnnf}s can be equivalently defined
as Free Binary Decision Diagrams (\textsc{fbdd}s) equipped 
with decomposable conjunction nodes \cite{BeameDNNF}. We refer to such that nodes 
as $\wedge_d$ nodes and to the resulting model as $\wedge_d$-\textsc{fbdd}. 
This is the way we are going to define Decision \textsc{dnnf}s. 

\begin{definition} \label{def:desdnnfsynt}
A $\wedge_d$-\textsc{fbdd} is a \textsc{dag} $B$ with one source and
at most two sinks. Each non-sink node has exactly two children. 
The labelling of nodes and edges as defined below.

If there are two sinks then one of them is labelled
with $True$, the other with $False$. If there is one sink, it can be labelled
either as $True$ or $False$. 

Each non-sink node is either labelled with a variable or with
$\wedge$, we thus refer to the former as a \emph{variable node} and to the latter 
as a \emph{conjunction node}.  We denote by $Var(B)$ the set of variables labelling the 
vertices of $B$, the $Var$ notation naturally extends to $B_u$ for $u \in  V(B)$. 
The labelling of nodes must observe the constraints of \emph{read-onceness}
and \emph{decomposability}. The read-onceness means that for in any directed
path there are no two distinct nodes labelled with the same variable.
The decomposability means that for any conjunction node $u$ with children $u_1$ and $u_2$,
$Var(B_{u_1}) \cap Var(B_{u_2})=\emptyset$. We thus refer to the conjunction nodes as 
$\wedge_d$-\emph{nodes} where $d$ in the subscript stands for 'decomposable'. 

Finally, one outgoing edge of each variable node is labelled with $0$, the other with $1$. 
\end{definition}

\begin{definition} \label{def:desdnnfsem}
Continuing on Definition \ref{def:desdnnfsynt}, we now define semantics of the 
$\wedge_d$-\textsc{fbbd} $B$. 
We first define a set of assignments $\mathcal{S}(B)$ over $Var(B)$. 
We then define $f(B)$, the function represented by $B$, as the function with $\mathcal{S}(f(B))=\mathcal{S}(B)$. 
Assume first that $B$ is a sink. If the node is labelled with $True$
then $\mathcal{S}(B)=\{\emptyset\}$. 
If the node is labelled with $False$ then $\mathcal{S}(B)=\emptyset$. 

In case $B$ has more than one node, let $u$ be the source of $B$
and let $u_0$ and $u_1$ be the children of $u$. 
Assume first that $u$ is labelled with a variable $x$. 
Let $(u,u_0)$ be labelled with $0$ and $(u,u_1)$ be labelled with $1$. 
Let $V_0=Var(B_{u_1}) \setminus Var(B_{u_0})$ and 
let $V_1=Var(B_{u_0}) \setminus Var(B_{u_1})$. 

We define $\mathcal{S}(B)=\{(x,0)\} \times \mathcal{S}(B_{u_0}) \times \{0,1\}^{V_0} \cup  
                          \{(x,1)\} \times \mathcal{S}(B_{u_1}) \times \{0,1\}^{V_1}$. 
Assume now that $u$ is a $\wedge_d$ node.
Then $\mathcal{S}(B)=\mathcal{S}(B_{u_0}) \times \mathcal{S}(B_{u_1})$. 													
\end{definition}

\begin{definition} \label{def:targetp1}
Let $B$ be a $\wedge_d$-\textsc{fbdd}.
Let $P$ be a path of $B$ starting  from the source. 
We refer to $P$ as a \emph{target path} of $B$. 

For a target path $P$, let  $Var(P)$ be the set of variables
labelling nodes of $P$ but the last one. 
Further on, let $\pi(P)$ the permutation of $Var(P)$
occurring in the order of their occurrence on $P$
starting from the source. 
Finally, let ${\bf a}={\bf a}(P)$ be the assignment of $Var(P)$ such that
for each $x \in var(P)$, ${\bf a}(x)$ is the label of $(u,v)$,
the edge of $P$ such that $u$ is labelled with $x$. 
\end{definition}

\begin{definition} \label{def:targetp2}
Let $B$ and $P$ be as in Definition \ref{def:targetp1}. 
We continue to introduce terminology related to target paths. 
First of all, let $u(P)$ be the final node of $P$. 
A \emph{junction} of $P$
is a $\wedge_{d}$-node of $P$ that is not $u(P)$. 
Let $u$ be a junction of $P$ and let $v$ be the child of $u$
that does not belong to $P$. We refer to $v$ as the \emph{alternative} for $u$. 
Also we say that $v$ is \emph{an alternative} for $P$ if $v$ is the alternative
for a junction of $u$. 

We denote by $Alt(P)$ the set of alternatives for $P$. 
We also let $V_0(P)=Var(B) \setminus (Var(B_{u(P)}) \cup \bigcup_{v \in Alt(P)} Var(B_v))$. 
\end{definition}

\begin{theorem} \label{th:decomp1}
With $P$ and $B$ as in Definition \ref{def:targetp1},
suppose that ${\bf a}={\bf a}(P)$ can be extended to a satisfying assignment of $f(B)$.
Then $V_0(P)$ are all not essential variables of $f(B)|_{\bf a}$. 
Moreover $\mathcal{S}(B)|_{\bf a}=\mathcal{S}(B_{u(P)}) \times \prod_{v \in Alt(P)} \mathcal{S}(B_v) \times \{0,1\}^{V_0(P)}$
if $|Alt(P)| \neq \emptyset$ and $\mathcal{S}(B)|_{\bf a}=\mathcal{S}(B_u) \times \{0,1\}^{V_0(P)}$
otherwise. 
\end{theorem}

\begin{proof}
Postponed to the Appendix. 
\end{proof}

Now, we introduce the terminology related to imbalanced $\wedge_d$-\textsc{fbdd}s.
\begin{definition}
Let $0<\alpha<1$. 
Let $B$ be a $\wedge_d$-\textsc{fbdd}. 
Let $n=|Var(B)|$. 
Let $u \in V(B)$. 
We say that $u$ is \emph{small}
$|Var(B_u)| \leq n^{\alpha}$. 
Otherwise, $u$ is \emph{large}

Let $u$ be a $\wedge_d$-gate and let $u_1$ and $u_2$ 
be the outgoing neighbours of $u$. 
We say that $u$ is $\alpha$-\emph{balanced}
if both $u_1$ and $u_2$ are large. 
Otherwise, the gate is $\alpha$-\emph{imbalanced}. 
We also refer to $\alpha$-imbalanced gates as $\wedge_{d,\alpha}$ ones.
We say that $B$ is $\wedge_{d,\alpha}$-\textsc{fbdd} if 
all the conjunction gates are $\wedge_{d,\alpha}$ ones. 
\end{definition}

From now on, we assume that $0<\alpha<1$ is a {\bf fixed constant}.
\begin{definition}
Let $B$ be a $\wedge_d$-\textsc{fbdd} and let $P$
be a target path of $B$. 
We say that $P$ is a \emph{mainstream path} of $B$
of all the alternatives of $P$ are small. 
\end{definition}

\begin{lemma} \label{lem:prelimconc}
Let $B$ be a $\wedge_d$-\textsc{fbdd} and let $P$
be a mainstream path of $B$. Suppose that
$\mathcal{S}(B)|_{{\bf a}(P)} \neq \emptyset$. 
Let $U \subseteq Var(B) \setminus Var(P)$ such that $|U|>max(n^{\alpha},2)$
(where $n=|var(B)|$) and ${\bf a}(P)$ does not break $U$. 
Then $u(P)$ is a large node and $U \subseteq Var(B_{u(P)})$. 
\end{lemma}

\begin{proof}
\begin{claim} \label{clm:prelimconc}
All the variables of $U$ are essential for $\mathcal{S}(B)|_{{\bf a}(P)}$
\end{claim}

\begin{proof}
Assume the opposite and let $u \in U$ be a variable that
is not essential for $\mathcal{S}(B)|_{{\bf a}(P)}$. 
Then $\mathcal{S}(B)|_{{\bf a}(P)}$ can be represented as a rectangle
$\mathcal{S} \times \{0,1\}^{\{u\}}$. As $|U|>2$, $Var(\mathcal{S}) \cap U \neq \emptyset$. 
We conclude that $\mathcal{S}(B)|_{{\bf a}(P)}$ breaks $U$ in contradiction to our assumption.
\end{proof}

It follows from the combination of 
Claim \ref{clm:prelimconc}
and Theorem \ref{th:decomp1}
that either $U \subseteq B_{u(P)}$
or $U \subseteq B_v$ for some $v \in Alt(P)$. 
As $v$ is small  due to $P$ being
mainstream,
the latter possibility can be ruled by assumption about the
size  of $U$. This completes the proof of the lemma.
\end{proof}

\begin{definition}
Let $B$ be a $\wedge_d$-\textsc{fbdd} and let ${\bf g}$ 
be  an assignment of $Var(B)$. 
Let $u \in V(B)$. 
We say that ${\bf g}$ is \emph{carried through} $u$
if there is a target path $P$ of $B$ with $u=u(P)$ 
such that ${\bf a}(P) \subseteq {\bf g}$. 
\end{definition}

\begin{lemma} \label{lem:prelimconc2}
Let $B$ be a $\wedge_d$-\textsc{fbdd} and let $u \in V(B)$ (that is, $u$ is a node of $B$). 
Let $x \in Var(B_u)$, $a \in \{0,1\}$, and assume that 
there is ${\bf h} \in \mathcal{S}(B_u)$ such that ${\bf h}(x)=a$. 

Further on, let $y \in Var(B) \setminus Var(B_u)$, $b \in \{0,1\}$,
and let ${\bf g} \in \mathcal{S}(B)$ be an assignment carried through $u$
such that ${\bf g}(y)=b$. 
Then $f(B)$ has a satisfying assignment assigning $x$ with $a$ and $y$ with $b$. 
\end{lemma}

{\bf Proof.}
Let ${\bf g'}=Proj({\bf g}, Var(B) \setminus Var(B_u))$. 
By assumption, $y \in Var({\bf g'})$ and hence ${\bf g'}(y)=b$. 
It follows from Theorem \ref{th:decomp1} that
${\bf g'} \cup {\bf h} \in \mathcal{S}(B)$. 
Clearly, ${\bf g'} \cup {\bf h}$ assigns $x$ and $a$ and $y$ with $b$. 
$\blacksquare$

\section{An XP lower bound for $\wedge_{d,\alpha}$-\textsc{fbdd} representing
CNFs of bounded treewidth} \label{sec:main}
\begin{theorem} \label{th:mainlowerbound}
Let $B$ be a $\wedge_{d,\alpha}$-\textsc{fbdd}
representing $\varphi(T[h,k])$. 
Then for a sufficiently large $h$ compared to $k$,
$|B| \geq n^{\Omega((1-\alpha) \sqrt{k}}$. 
\end{theorem}

Theorem \ref{th:mainlowerbound} is the main technical statement of this paper. 
In combination with Proposition \ref{prop:twup}, it implies Theorem \ref{th:mainintro}
as an immediate corllary. In the rest of this section we prove
Theorem \ref{th:mainlowerbound}.

\begin{definition}
Let $G$ be a graph. 
Let $W,U_0,U_1$ be disjoint subsets of $V(G)$. 
We call $(W,U_0,U_1)$ a \emph{target triple} if the following conditions hold. 
\begin{enumerate}
\item Each connected component of $G[U_1]$ is of size 
larger than $n^{\alpha}$. 
\item Each vertex of $U_0$ is adjacent to both $U_1$ and $W$.
\item $U_1$ is not adjacent to $W$. 
\item $U_0$ is an independent set. 
\item $N(U_0) \cap W$ is an independent set. 
\item Each vertex of $W$ is adjacent to at most one vertex of $U_0$. 
To put it differently, let $u_1,u_2$ be two distinct vertices of $U_0$
then $N(u_1) \cap W$ is disjoint with $N(u_2) \cap W$. 
\end{enumerate}
The \emph{rank} of $(W,U_0,U_1)$ is $|U_0|$. 
\end{definition}

The notion of target triple is a formalization of the intuition illustrated 
on the right-hand part of Figure \ref{fig:permstruct}. The green vertices represent $U_0$.
The red vertices belong to $W$, one per neighborhood of each green vertex. 
Finally, the constraint on $U_1$ as in the first item of the definition
plus adjacency of each vertex $U_0$ to $U_1$ clearly implies the same
constraint on $U_0 \cup U_1$. This ensures that each green vertex belongs to a heavy
`anchor' of size greater than $n^{\alpha}$ as depicted on the bottom part of
the right-hand side of Figure \ref{fig:permstruct}.


The following statement is the main combinatorial engine 
of the paper. It is proved in Section \ref{sec:maincomb}. 

\begin{lemma} \label{lem:largeindep1}
Let $\pi$ be a permutation of $V(T[h,k])$. 
Then there is a prefix $\pi_0$ of $\pi$
and disjoint subsets $U_0,U_1$ of $V(G) \setminus \pi_0$ 
such that $(\pi_0,U_0,U_1)$ is a target triple
of rank $\Omega((1-\alpha) \cdot h \sqrt{k})$.
\end{lemma}

In order to `harness' Lemma \ref{lem:largeindep1}
into the proof of Theorem \ref{th:mainlowerbound},
we need several auxiliary statements. 
The first one, provided below, effectively restates
Lemma \ref{lem:largeindep1} in terms of a branching 
program representing $\varphi(T[h,k])$. Its proof
is provided in Subsection \ref{sec:harness1}

\begin{lemma} \label{lem:largeindep3}
Let $B$ be a $\wedge_{d,\alpha}$-\textsc{fbdd}
representing $\varphi(T[h,k])$. 
Let ${\bf g}$ be a satisfying assignment 
of $\varphi(T[h,k])$. 
Then there are $u({\bf g}) \in V(B)$,
a mainstream path $P({\bf g})$ of $B$ and 
$U_0({\bf g}),U_1({\bf g}) \subseteq V(G) \setminus Var(P)$
subject  to the following. 
\begin{enumerate}
\item $u({\bf g})=u(P({\bf g}))$. 
\item $(Var(P({\bf g})),U_0({\bf g}),U_1({\bf g}))$ is a target
triple of rank $\Omega((1-\alpha) \cdot h \sqrt{k})$.
\end{enumerate}
\end{lemma}

The lower bound of Theorem \ref{th:mainlowerbound}
is established by observing that the set of nodes $u({\bf g})$
for all the satisfying assignments ${\bf g}$ of $\varphi(T[h,k])$
is large. This, in turn, is done by defining a probability space
over the satisfying assignments of $\varphi(T[h,k])$ and 
demonstrating that the probability that an assignment is carried 
through $u({\bf g})$ is small. Since each satisfying assignment 
is carried through some $u({\bf g})$, the total number of such 
nodes must be large due to the union bound. 

The remaining auxiliary statements enable the above reasoning 
by defining an independent set of size at least 
$\Omega((1-\alpha) \cdot h \sqrt{k})$ and showing 
that each satisfying assignment carried through $u({\bf g})$
assigns positively all the vertices of this independent set. 
The proof of one of these statements (Lemma \ref{lem:elpos})
is provided in Subsection \ref{sec:harness2}.

\begin{definition} \label{def:iset}
Let $G$ be a graph without  isolated vertices.
Let $B$ be a $\wedge_{d,\alpha}$-\textsc{fbdd}
representing $\varphi(G)$. 
Let $P$ be a mainstream path of $B$. 
Let $U_0,U_1 \subseteq V(G) \setminus Var(P)$ be such that
$(Var(P),U_0,U_1)$ is a target triple. 
For each $u \in U_0$, let $I_u(P,U_0,U_1)$ be defined as follows.
\begin{enumerate}
\item If ${\bf a}(P)$ assigns positively all  of $N(u) \cap W$
then $I_u(P,U_0,U_1)=N(u) \cap W$. 
\item Otherwise, $I_u(P,U_0,U_1)=N(u) \cap W=\{u\}$.
\end{enumerate}
We let $I(P,U_0,U_1)=\bigcup_{u \in U_0}$ 
\end{definition}

\begin{lemma} \label{lem:largeindep2}
With the notation as in Definition \ref{def:iset},
$I(P,U_0,U_1)$ is an independent set. 
\end{lemma}

\begin{proof}
Assume the opposite. 
We observe that for each $u \in U_0$,
$I_u(P,U_0,U_1)$ is independent. 
Indeed, if $I_u(P,U_0,U_1)$ is of size greater than
$1$ then $I_u(P,U_0,U_1)=N(u) \cap Var(P) \subseteq N(U_0) \cap Var(P)$
and the last set is independent by definition of a target triple. 
It follows that any edge must
be between $I_{u_1}(P,U_0,U_1)$ and $I_{u_2}(O,U_0,U_1)$
for distinct  $u_1,u_2 \in U_0$.
It follows from definition of target triple 
that say $I_{u_1}(P,U_0,U_1)=N_{u_1} \cap  Var(P)$ 
and $I_{u_2}(P,U_0,U_1)=\{u_2\}$. 
However, the adjacency of $u_2$ to $N_{u_1} \cap Var(P)$
is ruled out by the last item in the definition of a target triple.
\end{proof}

\begin{lemma} \label{lem:elpos}
With the notation as in Definition \ref{def:iset},
Assume that ${\bf a}(P)$ can be extended to a satisfying assignment of $\varphi(G)$
and let ${\bf g}$ be a satisfying assignment of $\varphi(G)$ carried 
through $u(P)$. 
Then ${\bf g}$ assigns positively all the elements of
$I(P,U_0,U_1)$.  
\end{lemma}

\begin{proof}
{(\bf of Theorem \ref{th:mainlowerbound}).}
Let ${\bf e}$ be a function on $E(T[h,k])$
mapping each edge $e$ to one of its ends.
That is, if $e=\{u,v\}$ then either  ${\bf e}(e)=u$
or ${\bf e}(e)=v$. Let $\mathcal{E}$ be the uniform 
probability space over such functions. 
To put it differently, we throw a fair coin for each
edge selecting one of the ends depending on the outcome. 
So, the probability of the result of throwing a coin for each
edge is $0.5$ to the power of $|E(T[h,k])|$.

For each ${\bf e} \in \mathcal{E}$,
let $Out({\bf e})=\{u| \exists e \in E(T_{h,k}), {\bf e}(e)=u\}$. 
For an assignment ${\bf g} \in \mathcal{S}(\varphi(T[h,k]))$, let $Pos({\bf g})$
be the subset of $V(T[h,k])$ consisting of all the elements that
are assigned with $1$ by ${\bf g}$. 
Let $Assign({\bf e})$ be the assignment ${\bf g}$ such that $
Pos({\bf g})=Out({\bf e})$. 
We thus also consider ${\bf g}$ as the event consisting 
of all ${\bf e}$ such that $Assign({\bf e})={\bf g}$. 
We let $Pr({\bf g})$ be the probability of this event.
Thus, we have defined a probability space over $\mathcal{S}(T[h,k])$. 
In what follows we will use only this probability space. 

For $S \subseteq V(T[h,k])$, 
let ${\bf Set}(S)$ be the event consisting of all ${\bf g}$
such that $S \subseteq Pos({\bf g})$. Since the maximum degree
of $T[h,k]$ is $6$, the following holds \cite{RazgonAlgo}.

\begin{claim} \label{clm:probmain}
There is a constant $\beta>1$ such that the following holds. 
Let $S$ be an independent set of $T[h,k]$. 
Then $Pr({\bf Set}(S)) \leq \beta^{-|S|}$
\end{claim} 

For each $u \in V(B)$, let ${\bf C}(u)$ be the event
consisting of all ${\bf g}$ such that 
$u({\bf g})=u$. We observe that 
$\bigcup_{u \in V(B)}({\bf C}(u))=\mathcal{S}(\varphi(T[h,k]))$. 
Indeed, clearly, the equality may be replaced by a subset.
On the other hand, a satisfying assignment ${\bf g}$ of $\varphi(T[h,k])$
belongs to ${\bf C}_{u({\bf g})}$ and hence to the union of the ${\bf C}$-sets. 
In light  of the union bound, we conclude that 

\begin{equation} \label{eq:mainprob}
\sum_{u \in V(B)} Pr({\bf C}_u) \geq 1
\end{equation}

We proceed to upper bound the probability of ${\bf C}(u)$ 
for a specific $u \in V(B)$. 
If for each ${\bf g} \in \mathcal{S}(\varphi(T[h,k]))$,
$u({\bf g}) \neq u$ then $Pr({\bf C}(u))=0$. 
Otherwise, fix ${\bf g} \in \mathcal{S}(\varphi(T[h,k]))$
such that $u({\bf g})=u$. Let ${\bf C^*}(u)$ be the 
even consisting of all ${\bf h}$ such that ${\bf h}$
is carried through $u$. 
Clearly, ${\bf C}(u) \subseteq {\bf C^*}(u)$, hence we
can upper-bound the latter. 

It follows  from Lemma \ref{lem:elpos},
that ${\bf C^*}(u) \subseteq {\bf Set}(I(P({\bf g}),U_0({\bf g}),U_1({\bf g})))$. 
It follows from combination of Lemma \ref{lem:largeindep2} and Claim \ref{clm:probmain}
that $Pr({\bf Set}(I(P({\bf g}),U_0({\bf g}),U_1({\bf g})))) \leq 2^{|I(P({\bf g}),U_0({\bf g}),U_1({\bf g}))|}$. 
By construction, $|I(P({\bf g}),U_0({\bf g}),U_1({\bf g}))| \geq |U_0|$
whereas by Lemma \ref{lem:largeindep3}, $|U_0| \geq \Omega((1-\alpha) \cdot h \sqrt{k})$.
We thus conclude that $Pr({\bf C^*}(u({\bf g}))) \leq n^{-\Omega((1-\alpha) \cdot k)}$. 
Combining the last equation with
\eqref{eq:mainprob},
we conclude that the number of distinct $u \in V(B)$ is at least
$n^{\Omega((1-\alpha) \sqrt{k})}$ as required.
\end{proof}

\subsection{Proof of Lemma \ref{lem:largeindep3}} \label{sec:harness1}
We define the path $P_0({\bf g})$ as the output of the algorithm below. 

\begin{enumerate}
\item $P \leftarrow ()$ and let $curnode$ be the source of $B$.
\item While $curnode$ is large
   \begin{enumerate}
	 \item If $curnode$ is a $\wedge_{d, \alpha}$-node
	       \begin{enumerate}
				 \item If $curnode$ has a large child let $newnode$
				       be this child (note that such a child is unique). 
							 Otherwise, let $newnode$ be an arbitrary child of $curnode$.
				 \end{enumerate}
	 \item Else
	       \begin{enumerate}
				 \item Let $x$ be the variable labelling $curnode$
				 \item Let $newnode$ be the child of $curnode$ such that
				       $(curnode,newnode)$ is labelled with ${\bf g}(x)$. 
				 \end{enumerate}
	 \item $P \leftarrow P.append(newnode)$
	       (that is, $newnode$ is added to the end of $P$), 
				 $curnode \leftarrow newnode$
	 \end{enumerate}
\item Return $P$. 
\end{enumerate}

Let $\pi_0({\bf g})=\pi(P_0({\bf g}))$. 
Fix an arbitrary permutation $\pi_1({\bf g})$ for which $\pi_0({\bf g})$ 
is a prefix. Then we apply Lemma \ref{lem:largeindep1}
to identify a prefix $\pi^*({\bf g})$ of $\pi_1({\bf g})$
and two subsets $U_0({\bf g})$ and $U_1({\bf g})$ 
of $V(G) \setminus \pi^*({\bf g})$ such that
$(\pi^*({\bf g}),U_0({\bf g}),U_1({\bf g}))$ is a target triple.

We claim that $\pi^*({\bf g})$ is a proper prefix of $\pi_0({\bf g})$. 
Note that the correctness of the claim implies the lemma. Indeed, we just
let $P({\bf g})$ be the prefix of $P_0({\bf g})$ such that 
$\pi(P({\bf g}))=\pi^*({\bf g})$ and observe that 
$(Var(P({\bf g})),U_0({\bf g}),U_1({\bf g}))$ is a desired triple.
It thus remains to prove the claim.

Assume the opposite. Then $\pi_0({\bf g})$ is a prefix (not necessarily proper)
of $\pi^*({\bf g})$. 
Let ${\bf g^*}=Proj({\bf g},\pi^*({\bf g}))$. 
Let $V^*$ be a connected component of $G[U_1({\bf g})]$ of size greater than $n^{\alpha}$. 
Since $U_1({\bf g})$ is not adjacent to $\pi^*({\bf g})$, we conclude that $V^*$
is unfixed by ${\bf g^*}$. 
Let ${\bf g}_0=Proj({\bf g},\pi_0({\bf g}))$. 
As, by assumption, ${\bf g}_0 \subseteq {\bf g^*}$, we conclude
that $V^*$ is unfixed by ${\bf g}_0$ either. 
We further note that ${\bf g}_0={\bf a}(P_0({\bf g}))$
It follows from Lemma \ref{lem:nobreak2} that 
$\mathcal{S}(B)|_{{\bf a}(P_0({\bf g}))}$ does not break $V^*$. 
By Lemma \ref{lem:prelimconc}, $V^* \subseteq Var(B_{u(P_0({\bf g}))})$.
However, this is a contradiction since, by construction, $u(P_0({\bf g}))$
is not a large node, while on the other hand $|V^*|>n^{\alpha}$ thus 
confirming the claim.

\subsection{Proof of Lemma \ref{lem:elpos}} \label{sec:harness2}


Let $Pos(P,U_0,U_1)$ be the subset of $U_0$ 
consisting of all vertices $u$ such that 
all of $N(u) \cap Var(P)$ are assigned positively by ${\bf a}(P)$. 
We denote $U_0 \setminus Pos(P,U_0,U_1)$ by $Neg(P,U_0,U_1)$. 
It is not hard to see that

\begin{equation} \label{eq:indep111}
I(P,U_0,U_1)=(N(Pos(P,U_0,U_1)) \cap Var(P)) \cup Neg(P,U_0,U_1)
\end{equation}

\begin{lemma} \label{lem:nobreak2}
\begin{enumerate}
\item $U_1 \cup Pos(P,U_0,U_1) \subseteq Var(B_{u(P)})$. 
\item For each $x \in U_1 \cup Pos(P,U_1)$,
$\mathcal{S}(B_{u(P)})$ contains an assignment that 
assigns $x$ negatively.
\end{enumerate}
\end{lemma}

\begin{proof}
It follows from the definition of a target triple
that $U_1 \cup Pos(P,U_0,U_1)$ is the union of connected 
components $V_1, \dots, V_q$ of $G$ of size greater than $n^{\alpha}$
each. Also each $V_i$ is unfixed by ${\bf a}(P)$. 
It follows from Lemma \ref{lem:nobreak2} that 
$\mathcal{S}(B)|_{{\bf a}(P)}$ does not break $V_i$. 
By Lemma \ref{lem:prelimconc}, $V_i \subseteq Var(B_{u(P)})$. 
This completes the proof of the first statement.
The second statement is immediate from Theorem \ref{th:decomp1} and Lemma \ref{lem:nobreak0}.
\end{proof}

\begin{proof}
{\bf of Lemma \ref{lem:elpos}.}
Let $u \in Pos(P,U_0,U_1)$. 
Towards a contradiction, assume existence of a satisfying 
assignment ${\bf g}$ of $\varphi(G)$ that is carried through 
$u(P)$ and assigns negatively some $v \in N(u) \cap Var(P)$. 
As $v \in Var(P)$, $v \notin Var(B_{u(P)})$ due to read-onceness. 
It follows  from combination of Lemma \ref{lem:nobreak2} and
Lemma \ref{lem:prelimconc2} that there is a satisfying assignment 
of $\varphi(G)$ that assigns negatively both $u$ and $v$ thus falsifying
the clause $(u \vee v)$.

Assume now that $u \in Neg(P,U_0,U_1)$. 
Towards a contradiction, we assume existence of a satisfying assignment ${\bf g}$
carried through $u(P)$ that assigns $u$ negatively. 
Assume first that $u \notin Var(B_{u(P)})$. 
Let $v \in N_u \cap U_1$ (existing by definition of a target triple). 
By Lemma \ref{lem:nobreak2}, 
there is an assignment of $\mathcal{S}(B_{u(P)})$ that
assigns $v$ negatively. 
It follows from Lemma \ref{lem:prelimconc2} that there is a satisfying assignment 
of $\varphi(G)$ that assigns negatively both $u$ and $v$ thus falsifying
the clause $(u \vee v)$.

It remains to assume that $u \in Var(B_{u(P)})$. 
Let $P'$ be a target path ending at $u(P)$ such that ${\bf a}(P') \subseteq {\bf g}$. 
By Theorem \ref{th:decomp1} applied to $P'$ there is ${\bf h} \in \mathcal{S}(B_{u(P')})$
that assigns $u$ negaitvely. 
On the other hand, since $u \in Neg(P,U_0,U_1)$, there is $w \in N(u) \cap Var(P)$
that is assigned negatively by ${\bf a}(P)$. Due to read-onceness,
$w \notin Var(B_{u(P)})$. By assumption about ${\bf a}(P)$ it can be 
extended to a satisfying assignment ${\bf a}$ of $\varphi(G)$. 
Clearly, this assignment is carried through $u(P)$. 
It follows from Lemma \ref{lem:prelimconc2}, that there is a satisfying assignment 
of $\varphi(G)$ that assigns negatively both $u$ and $w$ thus falsifying
the clause $(u \vee w)$. 
\end{proof}

\section{Proof of Lemma \ref{lem:largeindep1}} \label{sec:maincomb}

We start from terminology expansion, we will then provide an example to illustrate the 
new notions. 

The \emph{height} of a node $t$ denoted by $height(t)$ of $T[h]$ is its distance 
to a leaf. That is, the leaves have height $0$, their parents have 
height $1$ and so on. We stratify $V(T[h])$ according to heights of the nodes. 
In particular $V_a(T[h])$ is the subset of $T[h]$ consisting 
of all the nodes of height $a$. 
Next $V_{>a}(T[h])$ is the set of all the nodes of height greater than $a$. 
The sets $V_{<a}(T[h])$, $V_{ \geq a}(T[h])$ and $V_{\leq a}(T[h])$
are defined accordingly. 

We extend this terminology to $T[h,k]$. 
In particular, $V_a(T[h,k])=V_a(T[h]) \times [k]$. 
The sets $V_{>a}(T[h,k])$, $V_{<a}(T[h,k])$,
$V_{\leq a}(T[h,k])$, $V_{\geq a}(T[h,k])$ are defined 
accordingly. 

In order to introduce a more refined view of the above
sets, we shorten the notation. In particular,
we use $V(h)$ instead of $V(T[h])$. 
The shortenings $V_a(h)$, $V_{>a}(h)$, $V_{<a}(h)$,
$V_{\leq a}(h)$ and $V_{\geq a}(h)$ are defined accordingly. 
We introduce analogous shortenings for the subsets of $V(T[h,k])$. 
In particular, we denote $V(T[h,k])$ by $V(h,k)$. 
The shortenings $V_a(h,k)$, $V_{>a}(h,k)$, $V_{<a}(h,k)$,
$V_{\leq a}(h,k)$ and $V_{\geq a}(h,k)$ are defined accordingly.

Let $t \in V(h)$. We denote $V(T[h]_t)$ by $V(h,t)$. 
We denote $V(h,t) \cap V_a(h)$ by $V_a(h,t)$. 
To put it differently, $V_a(h,t)$ is the set of descendants 
of $t$ of height $a$. 
This context naturally leads to definition of sets $V_{>a}(h,t)$, $V_{<a}(h,t)$,
$V_{\geq a}(h,t)$  and $V_{\leq a}(h,t)$. 

The notation in the previous paragraph naturally extends to $T[h,k]$. 
In particular, we denote $V(h,t) \times [k]$ by $V(h,k,t)$. 
Next, we denote $V(h,k,t) \cap V_a(h,k)$ by $V_a(h,k,t)$. 
This context naturally leads to definition of sets $V_{>a}(h,k,t)$, $V_{<a}(h,k,t)$,
$V_{\geq a}(h,k,t)$  and $V_{\leq a}(h,k,t)$.

Now, we are moving towards defining families of sets as above. 
First let ${\bf S}_a(h)=\{V(h,t)| t \in V_a(h)\}$. 
Next, ${\bf S}_a(h,t)=\{V(h,t')|t' \in V_a(h,t)\}$.
These notions naturally extend to $T[h,k]$. 
In particular, ${\bf S}_a(h,k)=\{U \times [k]|U \in {\bf S}_a(h)\}$
and ${\bf S}_a(h,k,t)=\{U \times [k]| U \in {\bf S}_a(h,t)\}$. 

For each $U \in {\bf S}_a(h)$, $root(U)$ is $t \in V_a(h)$ such that
$U=V(h,t)$. 
Accordingly, for each $U \in {\bf S}_a(h,k)$, 
$root(U)$ is $t \in V(h)$ such that $U=V(h,k,t)$. 


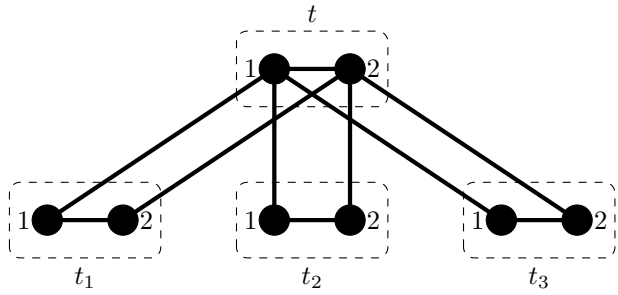
\begin{figure}[h]
\begin{tikzpicture}

\draw [fill=black]  (1,1) circle [radius=0.2];
\draw [fill=black]  (2,1) circle [radius=0.2];
\draw [dashed, rounded corners] (0.5,0.5) rectangle (2.5,1.5);
\node [left]  at (0.9,1)  {\bf $1$};
\node [right]  at (2.1,1)  {\bf $2$};
\node [below]  at (1.5,0.5)  {\bf $t_1$};

\draw [fill=black]  (4,1) circle [radius=0.2];
\draw [fill=black]  (5,1) circle [radius=0.2];
\draw [dashed, rounded corners] (3.5,0.5) rectangle (5.5,1.5);
\node [left]  at (3.9,1)  {\bf $1$};
\node [right]  at (5.1,1)  {\bf $2$};
\node [below]  at (4.5,0.5)  {\bf $t_2$};

\draw [fill=black]  (7,1) circle [radius=0.2];
\draw [fill=black]  (8,1) circle [radius=0.2];
\draw [dashed, rounded corners] (6.5,0.5) rectangle (8.5,1.5);
\node [left]  at (6.9,1)  {\bf $1$};
\node [right]  at (8.1,1)  {\bf $2$};
\node [below]  at (7.5,0.5)  {\bf $t_3$};

\draw [ultra thick] (1,1) --(2,1);

\draw [ultra thick] (4,1) --(5,1);

\draw [ultra thick] (7,1) --(8,1);

\draw [fill=black]  (4,3) circle [radius=0.2];
\draw [fill=black]  (5,3) circle [radius=0.2];
\draw [dashed, rounded corners] (3.5,2.5) rectangle (5.5,3.5);
\node [left]  at (3.9,3)  {\bf $1$};
\node [right]  at (5.1,3)  {\bf $2$};
\node [above]  at (4.5,3.5)  {\bf $t$};

\draw [ultra thick] (4,3) --(5,3);

\draw [ultra thick] (1,1) --(4,3);
\draw [ultra thick] (2,1) --(5,3);

\draw [ultra thick] (4,1) --(4,3);
\draw [ultra thick] (5,1) --(5,3);

\draw [ultra thick] (7,1) --(4,3);
\draw [ultra thick] (8,1) --(5,3);
\end{tikzpicture}
\caption{$T[1,2]$: a closer look}
\label{fig:treeproductperm}
\end{figure}
Let us consider several examples of the above terminology.
Figure \ref{fig:treeproductperm} illustrates $T[1,2]$ 
with the root of $T[1]$ being $t$ and children
$t_1,t_2,t_3$. The vertices of $P_k$ are $1,2$. 
First, $V_0(1)=\{t_1,t_3,t_3\}$, $V_{>0}(1)=\{t\}$. 
Further on, $V_0(1,2)=\{(t_1,1),(t_1,2),(t_2,1),(t_2,2),(t_3,1),(t_3,2)\}$
and $V_{>0}(1,2)=\{(t,1),(t,2)\}$. 
Next, ${\bf S}_0(1,2)=\{\{(t_1,1),(t_1,2)\},\{(t_2,1),(t_2,2)\},\{(t_3,1),(t_3,2)\}\}$
and for each $i \in [3]$, $root(\{(t_i,1),(t_i,2)\})=t_i$. 
Finally, ${\bf S}_{>0}(1)=\{\{t,t_1,t_2,t_3\}\}$.

\begin{definition}
Let $\pi$ be a permutation of $V(h,k)$. 
Let $h_0<h_1 \leq h$. 
We say that $\pi$ is $h_0,h_1$-\textsc{td} (where `td' stands for `top down')
if the following holds for each $t \in V_{h_1}(h)$. 
Let $\pi_0$ be the shortest prefix of $\pi$
such that $|\pi_0 \cap V_{>h_0}(h,k,t)| \geq k$. 
Then there is $U \in {\bf S}_{h_0}(h,k,t)$ such that
$\pi_0 \cap U=\emptyset$. 
\end{definition}

\begin{definition}
Let $\pi$ be a permutation of $V(h,k)$. 
Let $h_0<h_1 \leq h$. 
We say $\pi$ is $h_0,h_1$-\textsc{bu} (where `bu' stands for `bottom  up')
if there is $t \in V_{h_1}(h)$ such that there is a prefix 
$\pi_0$ of $\pi$ for which the following two statements hold. 
\begin{enumerate}
\item For each $U \in {\bf S}_{h_0}(h,k,t)$, $\pi_0 \cap U \neq \emptyset$. 
\item $\pi_0 \cap V_{>h_0}(h,k,t)=\emptyset$. 
\end{enumerate}
\end{definition}

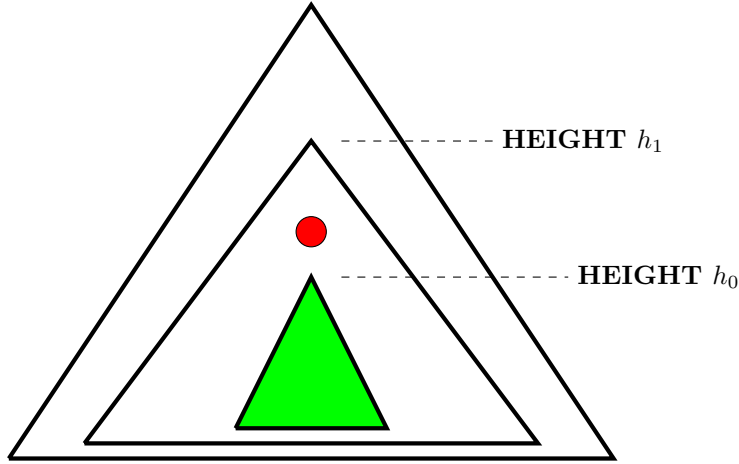
\begin{figure}[h]
\begin{tikzpicture}
\draw [ultra thick] (1,1) --(9,1)--(5,7)--(1,1); 
\draw [ultra thick] (2,1.2) --(8,1.2)--(5,5.2)--(2,1.2); 
\draw [ultra thick,fill=green] (4,1.4) --(6,1.4)--(5,3.4)--(4,1.4); 
\draw [fill=red]  (5,4) circle [radius=0.2];
\draw [dashed] (5.4,5.2)--(7.4,5.2); 
\draw [dashed] (5.4,3.4)--(8.4,3.4); 
\node [right]  at (7.4,5.2)  {\bf HEIGHT $h_1$};
\node [right]  at (8.4,3.4)  {\bf HEIGHT $h_0$};
\end{tikzpicture}
\caption{Top down permutation: an intuitive illustration}
\label{fig:permtd}
\end{figure}

\begin{figure}[h]
\begin{tikzpicture}
\draw [ultra thick] (1,1) --(9,1)--(5,7)--(1,1); 
\draw [ultra thick] (1.5,1.2) --(3.5,1.2)--(2.5,3)--(1.5,1.2); 
\draw [ultra thick] (4,1.2) --(6,1.2)--(5,3)--(4,1.2); 
\draw [ultra thick] (6.5,1.2) --(8.5,1.2)--(7.5,3)--(6.5,1.2); 
\draw [ultra thick,fill=green] (2.33,3) --(7.67,3)--(5,7)--(2.33,3); 
\draw [fill=red]  (2.5,2) circle [radius=0.2];
\draw [fill=red]  (5,2) circle [radius=0.2];
\draw [fill=red]  (7.5,2) circle [radius=0.2];
\draw [dashed] (5.4,7)--(7.4,7); 
\node [right]  at (7.4,7)  {\bf HEIGHT $h_1$};
\draw [dashed] (8,3)--(9,3); 
\node [right]  at (9,3)  {\bf HEIGHT $h_0$};
\end{tikzpicture}
\caption{Bottom up permutation: an intuitive illustration}
\label{fig:permbu}
\end{figure}
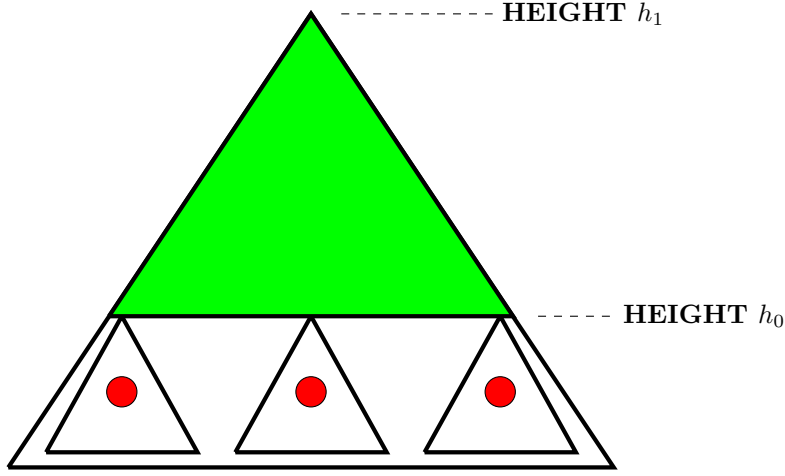

Figure \ref{fig:permtd} provides intuition for the 
notion of an $h_0,h_1$-\textsc{td} permutation. 
The inner white triangle depicts $V_{h_1}(h,k,t)$ (there may many such sets,
one for each $t \in V_{h_1}(h)$). 
The red circle inside the triangle denotes vertices of $\pi_0$. 
The green triangle at height $h_0$ denotes $U \in {\bf S}_{h_0}(h,k,t)$ such that
$\pi_0 \cap U=\emptyset$.

Intuition for the notion of a $h_0,h_1$-\textsc{bu} permutation is illustrated 
on Figure{fig:permbu}. The outermost trianlge on the picture 
depicts a set $V(h,k,t)$ for some $t \in V_{h_1}(h)$ that witnesses the properties
of $\pi_0$ as in the definition. In particular, the sets $U \in {\bf S}_{h_0}(h,k,t)$
are depicted as triangles at height $h_0$ and the red circle in each of them
means a non-empty intersection with $\pi_0$. 
The green triangle depicts $V_{>h_0}(h,k,t)$ and the green colour means 
the empty intersection with $\pi_0$. 

Note the asymmetry of the above two definitions.
To be $h_0,h_1-\textsc{td}$ something must hold for each $t \in V_{h_1}(h,k)$
whereas $h_0,h_1$-\textsc{bu} requires some conditions to hold just for a single 
$t \in V_{h_1}(h,k)$.

Below we demonstrate that these notions can be thought as dual. 

\begin{lemma} \label{lem:winwin}
Let $\pi$ be a permutation of $V(h,k)$. 
Let $h_0+ \log_3 k+1<h_1 \leq h$.
Suppose that $\pi$ is not $h_0,h_1$-\textsc{td}. 
Then $\pi$ is $h_0,h_2$-\textsc{bu} where 
$h_2=h_1-\lceil \log_3 k \rceil$.  
\end{lemma}

\begin{proof}
First of all, we observe that by assumption
about $h_1$, $h_2>h_0$ so the corresponding 
constraint in the definition of $h_0,h_2$-\textsc{bu}
is satisfied. 

By assumption of the lemma, there is $t_0 \in V_{h_1}(h,k)$
such that the following holds. 

Let $\pi_0$ be the shortest prefix of $\pi$
such that $\pi_0 \cap V_{>h_0}(h,k,t_0) \geq k$. 
Then $\pi_0$ intersects with each $U \in {\bf S}_{h_0}(h,k,t_0)$. 
Let $\pi_1$ be the immediate predecessor of $\pi_0$
(that is obtained from $\pi_0$ by removal of the last element). 
By the minimality assumption about $\pi_0$, the removed element 
belongs to  $V_{>h_0}(h,k,t_0)$. 
That is, $\pi_1 \cap V_{>h_0}(h,k,t_0)<k$ while 
$\pi_1$ still intersects each $U \in {\bf S}_{h_0}(h,k,t_0)$.

By definition of $h_2$, 
$|{\bf S}_{h_2}(h,k,t_0)| \geq k$. 
Therefore, by the pigeonhole principle, 
there is $t_1 \in V_{h_2}(h,k,t_0) \subseteq V_{h_2}(h,k)$
such that $\pi_1$ intersects with each
$U \in {\bf S}_{h_0}(h,k,t_1)$ and does not intersect with
$V_{>h_0}(h,k,t_1)$ thus confirming the lemma. 
\end{proof}


Let us now outline the plan for the rest of this section. 
We will first prove Lemma \ref{lem:largeindep1} under assumption
that $\pi$ is a $h_0,h_1$-\textsc{td} permutation for suitably
selected parameters $h_0$ and $h_1$. Then we prove the lemma under assumption
$\pi$ is a $h_0,h_2$-\textsc{bu} permutation where $h_2=h_1-\lceil log_3 k \rceil$. 
Lemma \ref{lem:winwin} will then imply Lemma \ref{lem:largeindep1} for each permutation
$\pi$ of $T[h,k]$ for sufficiently large $h$ and $k$. 
The proof under assumption that $\pi$ is a $h_0,h_1$-\textsc{td}
permutation is a consequence of Lemma \ref{lem:tripletopdown} which, in turn, follows
from the combinatorial statement of Theorem \ref{th:maintopdowncomb}. 
The part of the section, starting from the next paragraph and up to Theorem \ref{th:maintopdowncomb}
provides terminology and auxiliary statements required for the proof of the theorem. 
The proof of Lemma \ref{lem:largeindep1} under assumption that $\pi$ is a $h_0,h_2$-\textsc{bu}
permutation is a consequence of Lemma \ref{lem:triplebottomup}. 

\begin{definition}
Let $h_0<h$ and let 
$U \subseteq V(h,k)$.
Let $P$ be a path of $T[h,k]$. 
We say that $P$ is a $U,h_0$-path if we can
name one terminal vertex of $P$ $first(P)$,
the other $last(P)$ (if $P$ consists of a single
vertex then $first(P)=last(P)$) so that the following 
holds.
\begin{enumerate}
\item $V(P) \setminus last(P) \subseteq V_{>h_0}(h,k)$. 
\item $last(P) \in V_{h_0}(h,k)$. Let $W \in {\bf S}_{h_0}(h,k)$ 
such that $last(P) \in W$. We refer to $W$ as $anchor(P)$. 
\item $(V(P) \cup anchor(P)) \cap U=\emptyset$. 
\item $first(P)$ is adjacent to $U$ and $V(P) \setminus \{first(P)\}$
is not adjacent to $U$. 
\end{enumerate} 
\end{definition} 

We refer to the set of all $U,h_0$-paths as ${\bf P}(U,h_0)$
Further on for $t \in V(h)$, we let ${\bf P}(U,h_0,t)$
be the subset of ${\bf P}(U,h_0)$ consisting of all $P$
such that $V(P) \subseteq V(h,k,t)$ (note that this implies that
$anchor(P) \subseteq V(h,k,t)$). 

Let ${\bf P} \subseteq {\bf P}(U,h_0)$. 
We let $first({\bf P})=\{first(P)|P \in {\bf P}\}$

\begin{definition}
With the notation as above, ${\bf P}$ is \emph{independent}
if the following conditions hold. 
\begin{enumerate}
\item $first({\bf P})$ is independent. 
\item $N(first({\bf P})) \cap U$ is independent. 
\item Each vertex of $N(first({\bf P})) \cap U$ is
adjacent to exactly one vertex of $first({\bf P})$. 
\end{enumerate}
\end{definition}




\begin{definition}
Let $U \subseteq V(h,k)$.
Then $Index(U)=\{i | \exists t \in V(h), (t,i) \in U\}$,
$Bags(U)=\{t| \exists i \in [k] (t,i) \in U\}$. 
Further on, we define 
$Index_t(U)=\{i| (t,i) \in U\}$ and
$Bags_i(U)=\{t| (t,i) \in U\}$. 
\end{definition}

Further on, we extend our notation  about paths. 
Let $t_1,t_2 \in V(h)$. We denote the unique path
between $t_1$ and $t_2$ by $P(t_1,t_2)$. 
Further on, let $P=t_1, \dots, t_q$ be a path 
of $T[h]$. We denote the path $(t_1,i), \dots, (t_q,i)$
of $T[h,k]$ by $P \times i$. 
Finally, let $P=i, \dots j$ be a path of $P_k$. 
We denote the path $(t,i) , \dots, (t,j)$ by $t \times P$. 

\begin{lemma} \label{lem:indepindex1}
Let $U \subseteq V(h,k)$. 
Let $U_0 \subseteq V(h)$ be a connected set of $T[h]$.  
Let $W_0=U_0 \times [k]$. 
Let $i \in Index(U \cap W_0 \cap V_{>h_0}(h,k))$ and 
let $S \in {\bf S}_{h_0}(h,k)$ such that $S \subseteq W_0$
and $S \cap  U=\emptyset$. 
Then there is $P \in {\bf P}(U,h_0)$ such that 
$V(P) \subseteq W_0 \cap (V(h) \times i)$.  
\end{lemma}

\begin{proof}
Let $t_0=root(S)$. 
Let $t_1 \in Bags_i(U \cap W_0 \cap V_{>h_0}(h,k))$
be the closest to $t_0$. Let $Q_0=P(t_1,t_0) \times i$
and let $Q_1=Q_0 \setminus (t_1,i)$. Note that, since 
$Q_0$ consists of at least two vertices, $Q_1$ is non-empty. 
Due to the minimality of $t_1$, $V(Q_1) \cap U=\emptyset$. 

Let $first(Q_1)$ be the immediate of successor of $(t_1,i)$ in $Q_0$
and $last(Q_1)$ be the other terminal vertex of $Q_1$ (the two terminal vertices
may coincide). $first(Q_1)$ is adjacent to $U$, namely to $(t_1,i)$. 
However, there may be other vertices adjacent to $U$. 
Let $w$ be the last vertex of $Q_1$ adjacent to $U$ as the path is explored starting from 
$first(Q_1)$. Let $Q_2$ be the subpath of $Q_1$ between $w$ and $last(Q_1)$. 
Then $Q_2$ is a required path with $first(Q_2)=w$, $last(Q_2)=last(Q_1)$ and
$anchor(Q_2)=S$.
\end{proof}

\begin{lemma} \label{lem:indepindex2}
Let $U \subseteq V(h,k)$. 
Let $U_0 \subseteq V(h)$ be a connected set of $T[h]$. 
Let $W_0=U_0 \times [k]$.  
Let $S \in {\bf S}_{h_0}(h,k)$ such that $S \subseteq W_0$ and $S \cap U=\emptyset$. 
Let $I=Index(U \cap W_0 \cap V_{>h_0}(h,k))$. 

Then there is an independent ${\bf P} \subseteq {\bf P}(U,h_0)$ 
such that $V(P) \subseteq W_0$ for each $P \in {\bf P}$
and $|{\bf P}| \geq |I|/7$. 
\end{lemma}

\begin{proof}
Let $I_0$ be a subset of $I$ of size at least $|I|/7$
such that for every two distinct $i_1,i_2 \in I_0$,
$|i_1-i_2| \geq 4$. 
Such a set can be formed in a greedy way. 
Take an element $i$ into $I_0$ and then remove all 
$u$ itself and all $j$ such that $|i-j| \leq 3$. 
Clearly, at most $7$ elements will be removed at the price of taking one. 

Next for each $i \in I_0$, let $P_i$ be a path as 
in Lemma \ref{lem:indepindex2}. Let ${\bf P}=\{P_i|i \in I_0\}$. 
It remains to show that these paths are independent. 
To this end, we observe that the distance between any two distinct
elements of $first({\bf P})$ is at least $4$. 
It follows that $first({\bf P})$ is an independent set, 
two distinct vertices of $first({\bf P})$ are not adjacent to the 
same element and vertices of $U$ adjacent to distinct elements of $first({\bf P})$
are not adjacent either. It remains to note that two distinct vertices adjacent
to the same vertex of $first({\bf P})$ are not adjacent because $T[h,k]$
is a triangle-free graph. 
\end{proof}

\begin{definition}
Let $t_0,t_1 \in V(h)$, $i,j \in [k]$. 
We refer to the path $t_0 \times (i \dots,j)+P(t_0,t_1) \times j$
as $P(t_0,t_1,i,j)$. 
To put it differently, the path first moves from 
$(t_0,i)$ to $(t_0,j)$ within $t_0 \times [k]$
and then from $(t_0,j)$ to $(t_1,j)$ along $P(t_0,t_1) \times j$. 
\end{definition}

\begin{lemma} \label{lem:indepbag1}
Let $U \subseteq V(h,k)$. 
Let $U_0 \subseteq V(h)$ be a connected set of $T[h]$. 
Let $W_0=U_0 \times [k]$.  
Let $S \in {\bf S}_{h_0}(h,k)$ such that $S \subseteq W_0$ and $S \cap U=\emptyset$. 
Let $I=Index(U \cap W_0 \cap V_{>h_0}(h,k))$. 
Let $B=Bags(U \cap W_0 \cap V_{>h_0}(h,k))$. 
Let $j \in [k] \setminus I$ such that for each $a \in I$, $|j-a| \geq 2$.

Let $t_0=root(S)$. 
Let $t_1 \in B$.
Let $i \in [k]$ such that $(t_1,i) \in U$ and
$i$ is closest to $j$. 
Then there is $P \in {\bf P}(U,h_0)$ 
that is a proper subpath of $P(t_1,t_0,i,j)$ and
that $first(P) \in t_1 \times [k]$.
\end{lemma}

\begin{proof}
Let $P_0=P(t_1,t_0,i,j) \setminus (t_1,i)$. 
Let $first(P_0)$ be the immediate successor of $(t_1,i)$ in $P(t_1,t_0,i,j)$
and let $last(P_0)$ be the other terminal vertex of $P_0$. 
We observe that $V(P_0) \cap U=\emptyset$. 
Indeed, $V(P_0) \cap U \cap (t_1 \times [k])=\emptyset$ by the minimality of $i$. 
The remaining vertices of $P_0$ are of the form $(t',j)$ and are, by selection
of $j$, cannot belong to $U$ (whose vertices are of the form $(t'',a)$, $a \in I$). 
Also, $first(P_0)$ is adjacent to $U$, namely to $(t_1,i)$. 
Arguing as in Lemma \ref{lem:indepindex1}, 
we take the last vertex $w$ of $P_0$ (being explored from $first(P_0)$)
that is adjacent to $U$. The suffix of $P_0$ starting from $w$ is 
an element of ${\bf P}(U,h_0)$. 
We observe that the vertices with the second coordinate $j$ are not adjacent 
to $U$. Hence the first vertex of the resulting path belongs to $t_1 \times [k]$. 
\end{proof}

\begin{lemma} \label{lem:indepbag2}
Let $U \subseteq V(h,k)$. 
Let $U_0 \subseteq V(h)$ be a connected set of $T[h]$. 
Let $W_0=U_0 \times [k]$.  
Let $S \in {\bf S}_{h_0}(h,k)$ such that $S \subseteq W_0$ and $S \cap U=\emptyset$. 
Let $I=Index(U \cap W_0 \cap V_{>h_0}(h,k))$. 
Let $B=Bags(U \cap W_0 \cap V_{>h_0}(h,k))$. 

Suppose that $|I| \leq \sqrt{k}$. 
Then, for a sufficiently large $k$, 
there is an independent ${\bf P} \subseteq  {\bf P}(U,h_0)$ 
such that $\bigcup_{P \in {\bf P}} V(P) \subseteq W_0$ 
and $|{\bf P}| \geq min(|B|/65, \sqrt{k})$.
\end{lemma}

\begin{proof}
We start from observing that there is a subset $B_0$ of $B$ 
of size at least $|B|/65$ that is an independent set of $T[h]$. Indeed, such a subset can be formed in a greedy way. 
Take an element $t' \in B$ and discard all the elements of $B$ at distance at most $3$ from $t'$. 
Since the max-degree of $T[h]$ is $4$, the number of discarded elements is $1+4^3=65$ (the extra $1$
is on account of $t'$). That is, we take one element to $B_0$ at the price of removal of at least $65$ elements 
from $B$. 

If $B_0 \leq \sqrt{k}$, let $B_1=B_0$. Otherwise, let $B_1$ be an arbitrary subset of 
$B_0$ of size $\lceil \sqrt{k} \rceil$. 
Let $ind_0$ be an injective mapping mapping from $B_1$ to $[k] \setminus I$
such that for any two distinct $t_1,t_2 \in B_1$,
$|ind_0(t_1)-ind_0(t_2)| \geq 4$ and such that for each $t' \in B_1$, $ind_0$ is at distannce at least $2$ from $I$.

We observe that such a mapping is possible for a sufficiently large $k$. 
Indeed, remove from $[k] \setminus I$ the elements with indicees adjacent to $I$. 
This will leave us with a set $I_0$ of size at least $k-3\sqrt{k}$. 
Out of this set form greedily a subset $I_1$ of elements at distance at least $4$
from each other. Arguing as in Lemma \ref{lem:indepindex2}, there is such a set
of size at least $(k-3\sqrt{k})/7$. 
For a sufficiently large $k$ the size of $I_1$ is greater than $\lceil \sqrt{k} \rceil$ and 
hence of $B_1$. Hence, the required mapping exists.

Next, let $ind_1$ be a mapping from $B_1$ to $[k]$ such that
for each $t_1 \in B_1$, $(t_1,ind_1(t_1)) \in U$ and $ind_1(t_1)$
is closest possible to $ind_0(t_1)$. 
Let $t_0=root(S)$. Let $P(t_1)$ be the subpath of $P(t_1,t_0,ind_1(t_1),ind_0(t_1))$
as specified in Lemma \ref{lem:indepbag1}. We let ${\bf P}=\{P(t_1)|t_1 \in B_1\}$. 
It remains to observe that ${\bf P}$ is independent. 

Let $t_1,t_2 \in B_1$. 
By Lemma \ref{lem:indepbag1}, $first(P(t_1)) \in t_1 \times [k]$
and $first(P(t_2)) \in t_2 \times [k]$. 
By definition of $B_1$, the distance between $first(P_1)$ and $first(P_2))$
is at least $4$. The remaining part of the argument is analogous to that of
Lemma \ref{lem:indepindex2}. 
\end{proof}

\begin{lemma} \label{lem:indepset0}
Let $U \subseteq V(h,k)$. 
Let $U_0 \subseteq V(h)$ be a connected set of $T[h]$. 
Let $W_0=U_0 \times [k]$.  
Let $S \in {\bf S}_{h_0}(h,k)$ such that $S \subseteq W_0$ and $S \cap U=\emptyset$. 
Let $W_1=U \cap W_0 \cap V_{>h_0}(h,k)$.
Suppose that $|W_1| \geq k$. 

Then, for a sufficiently large $k$, there is an independent ${\bf P} \in {\bf P}(U,h_0)$ 
with   $\bigcup_{P \in {\bf P}} V(P) \subseteq W_0$ 
and $|{\bf P}| \geq \sqrt{k}/65$.
\end{lemma}

\begin{proof}
Let $I=Index(W_1)$. 
and let $B=Bags(W_1)$.
Assume first that $|I| \geq \sqrt{k}$.  
Then the statement is immediate from Lemma \ref{lem:indepindex2}
(in fact with $\sqrt{k}/7$ being the lower bound). 

Otherwise, we observe that $W_1 \subseteq B \times I$
and hence $|W_1| \leq |B| \cdot |I|$. 
Hence, $|B| \geq \sqrt{k}$ and the lemma is immediate 
from Lemma \ref{lem:indepbag2}. 
\end{proof}

\begin{theorem} \label{th:maintopdowncomb}
Let $c(x_0,x_1)=1/65 \cdot (\lfloor (x_1-x_0)/4 \rfloor+1)$.
Let $0<h_0 <h_1 \leq h$ and suppose that $k$ is sufficiently large. 
Let $\pi$ be a permutation that is $(h_0,h_1)$-\textsc{td}.
Let $t \in V_{\geq h_1}(h)$. 
Then there is a prefix $\pi_0$ of $\pi$ 
such that the following two statements hold
\begin{enumerate}
\item There is an independent
set ${\bf P} \subseteq {\bf P}(\pi_0,h_0, t)$ of size 
at least $c(h_1,height(t)) \sqrt{k}$. 
\item There is $S \in {\bf S}_{h_0}(h,k,t)$ such that $\pi_0 \cap S=\emptyset$.
\end{enumerate}
In particular, there is a prefix $\pi_0$ of $\pi$ and
an independent set ${\bf P} \subseteq {\bf P}(\pi_0,h_0)$
of size $1/260 \cdot (h-h_1) \cdot \sqrt{k}$. 
\end{theorem}

\begin{proof}
By induction on $height(t)$. 
Assume first that $height(t)=h_1$. 
Let $W_0=V(h,k,t)$. 
By definition of a $(h_0,h_1)$-\textsc{td}
permutation, there is a prefix $\pi_0$ 
of $\pi$ such that $|W_0 \cap \pi_0 \cap V_{>h_0}(h,k,t)| \geq k$  
and there is $S \in {\bf S}_{h_0}(h,k,t)$ such that $S \cap \pi_0=\emptyset$. 
By Lemma \ref{lem:indepset0}, 
there is an independent set ${\bf P} \subseteq {\bf P}(\pi_0,h_0)$
of size at least $\sqrt{k}/65$ such that $\bigcup_{P \in {\bf P}} V(P) \subseteq W_0$. 
By definition of $W_0$, ${\bf P} \subseteq {\bf P}(\pi_0,h_0,t)$ thus 
confirming the theorem for the considered case.  

For the main case, let $t$ be be a node
with $height(t)>h_1$ and $height(t)-h_1$ being a multiple of $4$. 
Let $t^*_1,t^*_2,t^*_3$ be the children of $t$. 
Let $t_1,t_2,t_3$ be descendants of $t$ of height $height(t)-4$ such that
for each $i \in [3]$, $t_i$ is a descendant of $t^*_i$. 

By the induction assumption, for each $i \in [3]$, there is a prefix $\pi_i$
of $\pi$ such that there is an independent ${\bf P}_i \subseteq {\bf P}(\pi_i,h_0,t_i)$ 
of size $c(h_0,height(t_i)) \sqrt{k}$. 
We assume w.l.o.g. that $\pi_1 \subseteq \pi_2 \subseteq \pi_3$. 

Let $U_0=V(h,t) \setminus V(h,t^*_2)$. 
In other words, $U_0$ is the set of vertices of $T[h]_t \setminus T[h]_{t^*_2}$. 
Let $W_0=U_0 \times [k]$. 

\begin{claim} \label{clm:engine1}
$|W_0 \cap \pi_2 \cap V_{>h_0}(h,k,t)| \geq k$.
\end{claim}

\begin{proof}
By assumption, 
$W_0 \cap \pi_2 \cap V_{>h_0}(h,k,t)$ is a superset 
of $W_0 \cap \pi_1 \cap V_{>h_0}(h,k,t_1)$ and the size of the latter 
set is at least $k$.
\end{proof}

\begin{claim} \label{clm:engine2}
There is $S \in {\bf S}_{h_0}(h,k,t)$ such that
$S \subseteq W_0$ and $\pi_2 \cap S=\emptyset$. 
\end{claim}

\begin{proof}
Let $S \in {\bf S}_{h_0}(h,k,t_3)$
such that $\pi_3 \cap S=\emptyset$ such an $S$ exists by the induction assumption. 
Clearly $S \subseteq W_0$. 
On the other hand, as $\pi_2 \subseteq \pi_3$,
$\pi_2 \cap S=\emptyset$. 
\end{proof}

We observe that the premises of Lemma \ref{lem:indepset0} are satisfied
with $U=\pi_2$. Therefore, it follows from the lemma that there is a
a set ${\bf P'} \subseteq {\bf P}(\pi_2, h_0)$ of size at least $\sqrt{k} /65$
such that $\bigcup_{P \in {\bf P'}} V(P) \subseteq W_0$. 

Let ${\bf P}={\bf P}_2 \cup {\bf P'}$. 
By definition, ${\bf P} \subseteq {\bf P}(\pi_2,h_0)$
and $\bigcup_{P \in {\bf P}} V(P) \subseteq V(h,k,t)$. 
We are going to demonstrate that $|{\bf P}| \geq c(h_0,height(t)) \cdot \sqrt{k}$. 
and that ${\bf P}$ is independent. 

First of all, we notice that ${\bf P}_2$ and ${\bf P'}$ are disjoint simply
because the vertices of the former belong to $V_{h,k,t_2}$ where the vertices 
of the latter belong to $W_0$ which is disjoint with $V_{h,k,t_2}$. 

It follows that $|{\bf P}|=|{\bf P}_2|+|{\bf P'}| \geq 
c(h_0,height(t_2)) \sqrt{k}+1/65 \cdot \sqrt{k}$. 
Let us have a closer look $c(h_0,height(t_2))$. 
By assumption, $height(t_2)=4a$ for some integer $a$. 
Therefore, $c(h_0,height(t_2))=1/65 \cdot (a+1)$. 
It follows that $|{\bf P}| \geq 1/65 \cdot (a+2) \sqrt{k}=c(h_0,4(a+1)) \sqrt{k}$. 
However, $4(a+1)$ is exactly $height(t)$. Hence, the required lower bound on the 
size of ${\bf P}$ holds. 

It remains to verify the independence of ${\bf P}$.
Since ${\bf P}_2$ and ${\bf P'}$ are known to be independent.
we need to verify independence of a path from ${\bf P}_2$ and a 
a path from ${\bf P'}$. For this,  we observe that the distance between $V(h,k,t_2)$
and $W_0$ is at least $4$ and then apply the argument as in Lemma \ref{lem:indepindex2}. 

We thus have completed the proof of the first statement of the theorem
treating $\pi_2$ as $\pi_0$. The second statement is immediate from
Claim \ref{clm:engine2}. 

It remains to consider the case where $height(t)$ is not a multiple of $4$. 
Let $t_0$ be a closest to $t$ descendant of $t$ whose height is a mutiple of $4$. 
Let $\pi_0$ and ${\bf P} \subseteq {\bf P}(p_0,h_0,t_0)$ be the witnessing prefix and set of paths for $t_0$
existing by the induction assumption.
Since ${\bf P}(\pi_0,h_0,t_0) \subseteq {\bf P}(\pi_0,h_0,t)$
and $c(h_0,height(t_0))=c(h_0,t)$, the statements of the theorem hold regarding $t$.  
\end{proof}

Now we turn into transformation of $\pi_0$ and the corresponding set ${\bf P}$
of independent $\pi_0,h_0$-paths into a target triple. 

Denote $|V(h)|$ by $m(h)$. 
We assume that $m(h_0-1)>m(h)^{\alpha}$. 
For $u \in V(h,k)$, we introduce $bag(u)$ and $index(u)$
so that $u=(bag(u),index(u))$. 
We define $height(u)$ as $height(bag(u))$.

Let $P \in {\bf P}(U,h_0)$. 
Let $t_0(P)=bag(last(P))$. Note that $t_0(P)=root(anchor(P))$. 
Fix $t_1(P)$ as an arbitrary child of $t_0(P)$. 
Let $shift(P)$ be the path obtained from $P$ by appending 
$(t_1(P)),index(last(P))$ to the end of $P$ and removal of $first(P)$. 
Let $U_1(P)=V(shift(P)) \cup V(h,k,t_1)$. 

\begin{lemma} \label{lem:tripletopdown}
Let ${\bf P} \subseteq {\bf P}(U,h_0)$ be non-empty and independent. 
Let $U_0({\bf P})=\{first(P)| P \in {\bf P}\}$. 
Let $U_1({\bf P})=\bigcup_{P \in {\bf P}} U_1(P)$. 
Then $(U,U_0({\bf P}),U_1({\bf P}))$ is a target triple. 
\end{lemma}

\begin{proof}
Let $P \in {\bf P}$. 
We observe that $U_1(P)$  is a connected subset of $V(h,k)$. 
Moreover, $V(h,k,t_1(P)) \subseteq U_1(P)$. 
Since $height(t_1(P))=h_0-1$, $|V(h,t_1)|>m(h)^{\alpha}$. 
Hence $|V(h,k,t_1(P))|=|V(h,t_1(P))| \cdot k>m^{\alpha} \cdot \geq (mk)^{\alpha}=n^{\alpha}$. 
We thus conclude that $U_1({\bf P})$ is the union of connected sets of size 
greater than $n^{\alpha}$ each. 

We next observe that $U_1({\bf P})$ is not  adjacent with $U$ and does not intersect with $U$. 
It is enough to prove this for $U_1(P)$ for an arbitrary $P \in {\bf P}$. 
First of all, we observe that $shift(P) \cap V(P)$ is not adjacent with $U$.
This is immediate from definition of ${\bf P}(U,h_0)$ since 
$first(P) \notin shift(P) \cap V(P)$.
It thus remains to verify this for $V(h,k,t_1(P))$. 
The set does not intersect with $U$ since it is a subset of $anchor(P)$. 
The only outside neighbours of this set are the elements of 
$root(anchor(P)) \times [k]$ but they are not contained in $U$ again 
due to disjointness of $anchor(P)$ and $U$. 

Next, we observe that each element $u \in U_1({\bf P})$
is adjacent to both $U$ and $U_1({\bf P})$
Indeed, by construction, there is $P \in {\bf P}$ such that
$u=first(P)$ and $first(P)$ is adjacent to $U_1(P) \subseteq U_1({\bf P})$
simply by construction. Also $first(P)$ is adjacent to $U$ simply
since $P \in {\bf P}(U,h_0)$. 

The remaining properties of a target triple are satisfied 
due to the independence of ${\bf P}$. In particular, 
it implies the independence of $U_0({\bf P})$, the independence of $N(U_0(({\bf P})) \cap U$ and also that
each element of the latter set is adjacent ot exactly one element of $U_0({\bf P})$.
\end{proof}

We are now going  to define a construction for establishing
existence of a required target triple when $\pi$ is a bottom-up
permutation. 

With the assumption on $h_0$ as above, we also assume that $h_1$
is large enough compared to $h_0$. 

\begin{definition}
Let $t \in V_{h_1}(h)$ and let $U \subseteq V(h,k)$. 
We say that $U$ is \emph{easy on} $t$ if the following conditions hold. 
\begin{enumerate}
\item $U \cap V_{>h_0}(h,k,t)=\emptyset$. 
\item There is precisely one $S \in {\bf S}_{h_0+1}(h,k,t)$
such that  $U \cap S=\emptyset$. 
\end{enumerate}
\end{definition}

With $h_0$ and $h_1$ considered fixed, we denote by ${\bf Easy}(t)$
the family of all subsets of $V(h,k)$ that are easy on $t$. 
For each $U \in {\bf Easy}(t)$, 
we denote by $Free(U)$, the set $S \in {\bf S}_{h_0+1}(h,k,t)$
such that  $U \cap S=\emptyset$.
Further on, $root(U)$ is the child of $t$ that is an ancestor 
of $root(free(U))$. 
We also denote the family ${\bf S}_{h_0+1}(h,k,root(U)) \setminus \{Free(U)\}$
by ${\bf Loaded}(U)$.

\begin{figure}[h]
\begin{tikzpicture}
\draw [ultra thick] (1,1) --(9,1)--(5,7)--(1,1); 
\draw [ultra thick] (1.5,1.2) --(3.5,1.2)--(2.5,3)--(1.5,1.2); 
\draw [ultra thick] (4,1.2) --(6,1.2)--(5,3)--(4,1.2); 
\draw [ultra thick,fill=green] (6.5,1.2) --(8.5,1.2)--(7.5,3)--(6.5,1.2); 
\draw [fill=red]  (2.5,1.5) circle [radius=0.2];
\draw [fill=red]  (5,1.5) circle [radius=0.2];
\draw [fill=green]  (2.5,2.5) circle [radius=0.2];
\draw [fill=green]  (5,2.5) circle [radius=0.2];
\draw [fill=green]  (5,6.5) circle [radius=0.2];
\draw [green, ultra thick] (2.5,2.5)--(5,6.5); 
\draw [green, ultra thick] (5,2.5)--(5,6.5); 
\draw [green, ultra thick] (7.5,3)--(5,6.5); 
\draw [ultra thick] (2.5,1.7)--(2.5,2.3); 
\draw [ultra thick] (5,1.7)--(5,2.3); 

\end{tikzpicture}
\caption{Target triple for a bottom up permutation: an intuitive illustration}
\label{fig:ttbu}
\end{figure}
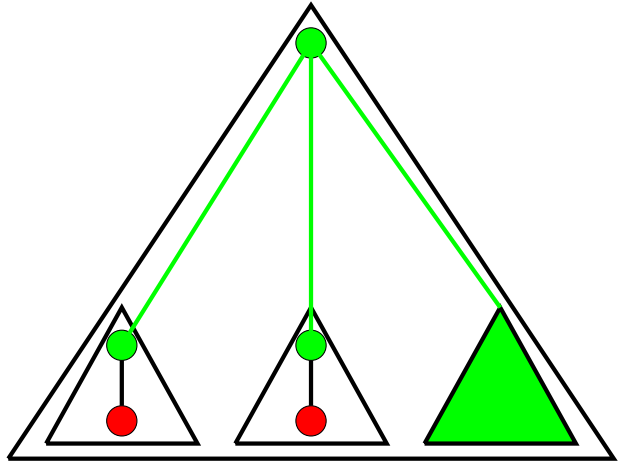

For each $S \in {\bf Loaded}(U)$, we let $w(S)$ to be an element of $S \cap U$
such that $height(w(S))$ is highest possible. 
We further define $P(S)$ as $(P(bag(w(S)),root(U)) \times index(w(S))) \setminus w(S)$.
Define $first(S)$ as the vertex of $P(S)$ adjacent to $w(S)$.
We denote the other terminal vertex of $P(S)$ as $last(S)$. 
We let $P^*(S)=P(S) \setminus first(S)$. 
Since we allowed $h_1-h_0$ to be sufficiently large, we assume 
that $last(S) \neq first(S)$ and, in particular, that $P^*(S)$ is non-empty.  
We let $P^*(U)=P(root(U),root(Free(U)) \times 1$. 
Let $U_0(U)=\{first(S)|S \in {\bf Loaded}(U)\}$
Let $U_1(U)=root(U) \times [k] \cup Free(U) \cup V(P^*(U)) \cup \bigcup_{S \in {\bf Loaded}(U)} V(P^*(S))$.

The above terminology is schematically illustrated on Figure \ref{fig:ttbu}. 
The outermost  triangle is $V(h,k,root(U))$ for some $U \in {\bf Easy}(t)$
and $t \in V_{h_1}(h)$. The inner rectangles depict ${\bf S}_{h_0+1}(h,k,t)$,
the green one is $Free(U)$, the remaining ones are elements of ${\bf Loaded}(U)$. 
The red circles inside elements $S \in {\bf Loaded}(U)$ depict $w(S)$,
the corresponding green circles depict the vertices $first(S)$. 
The green circle at the top of the outermost rectangle corresponds to the set $root(U) \times [k]$. 
The green lines between the top green circle and the bottom ones correspond to the paths
$P^*(S)$. The green line between the top green circle and the green rectangle corresponds
to $P^*(U)$. 
Thus the bottom green rectangles correspond to the elements of $U_0(U)$ and the union of the remaining
green elements is $U_1(U)$. The picture makes clear that $U_1(U)$ is a single connected component
that is large because it is `anchored' by $Free(U)$, that $U_0(U)$ is an independent set and that $U_0(U)$ is adjacent to both $U_1(U)$ and to $U$ (the latter through the lines between the green and 
red circles). 

We note an abuse of notation in the above terminology. For example, $Free(U)$ should, in fact, be 
$Free_t(U)$. Similar subscripts are omitted in the subsequent terminology for the sake of better readability. The correct use will always be clear from the context.

The intuition illustrated on Figure \ref{fig:ttbu} is made precise in the next lemma.

\begin{lemma} \label{lem:triplebottomup}
$(U,U_0(U),U_1(U))$ is a target triple. 
\end{lemma} 

\begin{proof}
First of all, we observe that $U_1(U)$ is connected.
Indeed, the vertices of paths $P^*(S)$ for $S \in {\bf Loaded}(U)$
all intersect a connected set $root(U) \times [k]$ which, 
in turn, intersects the connected set $V(P^*(U))$ which, again in turn,
intersects a connected set $Free(U)$.
Also $|U_1(U)| \geq |Free(U)|=m(h_0+1) \cdot k>n^{\alpha}$
by selection of $h_0$ and the same  argument as in the first paragraph
of the proof of Lemma \ref{lem:tripletopdown}. 

We then observe that $(U_0(U) \cup U_1(U)) \cap U=\emptyset$. 
Indeed, assume the opposite and let $u \in (U_0(U) \cup U_1(U)) \cap U$. 
We note that $u \subseteq V(h,k,root(U))$ and hence to belong to $U$,
$u \in S$ for some $S \in {\bf Loaded}(U)$. 
Then $u \in V(P(S))$. But then, by construction $height(u)>height(w(S))$
and the latter is the greatest height of an element of $U \cap S$.
This contradiction establishes that indeed $(U_0(U) \cup U_1(U)) \cap U=\emptyset$.

The most interesting part of the proof is establishing 
the non-adjacency of $U_1(U)$ and $U$ as this is the part of the target triple
definition that led us to introduce the notion of $root(U)$ and to use 
$V(h,k,root(U))$ instead of $V(h,k,t)$. 
First of all, we observe that $N(U_1(U)) \cap U \subseteq V(h,k,root(U))$. 
Indeed, $N(U_1(U)) \setminus V(h,k,root(U)) \subseteq t \times [k]$
and $t \times [k]$ is disjoint with $U$.  
Intuitively speaking, by using $V(h,k,root(U))$ instead of $V(h,k,t)$,
we achieve non-adjacency with $U \setminus V(h,k,t)$. 

Let $u \in N(U_1(U)) \cap U$. 
In light of the previous paragraph, there is $S \in {\bf Loaded}(U)$
such that $u \in S$. Let $v$ be a vertex of $U_1(U)$ adjacent to $u$. 
We note that, by definition of $U$, $height(u) \leq h_0$
whereas $height(root(S))=h_0+1$. We conclude that $v \in S$. 
It only possible if $v \in V(P^*(S))$. 
By construction, $height(v) \geq height(w(S))+2 \geq height(u)+2$,
the last inequality follows from the maximality of $w(S)$. 
This is a contradiction since the heghts of $u$ and $v$ can differ by at most $1$.

As each vertex of $U_0(U)$ is adjacent to $U$ simply by construction,
we also observe that $U_0(U)$ and $U_1(U)$ are disjoint. 
We also note that each first vertex $u \in U_0(U)$
is adjacent to $U_1(U)$. Indeed, by construction, $u=first(P(S))$
for some $S \in {\bf Loaded}(U)$. 
But then $u$ is adjacent to its successor in $P$ which is part of $U_1(U)$
again by construction. 

Let $u_1,u_2$ be two distinct vertices $U_0(U)$
By definition, there are two distinct $S_1,S_2 \in {\bf Loaded}(U)$
such that $u_1 \in S_1$ and $u_2 \in S_2$. 
As $S_1$ and $S_2$ are two distinct elements of ${\bf S}_{h_0+1}(h,k,root(U))$,
they are not adjacent. It follows that $u_1$ and $u_2$ are not adjacent either
and hence $U_0(U)$ is an independent set. 

Let $w_1$ and $w_2$ be two distinct elements of 
$N(U_0(U)) \cap U$. As we observed before, there are $S_1,S_2 \in {\bf Loaded}(U)$ 
such that $w_1 \in S_1$ and $w_2 \in S_2$. If $S_1 \neq S_2$ then $w_1$
and $w_2$ are not adjacent due to the same argument as in the previous paragraph.
Otherwise, they are both adjacent to $first(S_1)$ and hence non-adjacent between 
themselves since $T[h,k]$ is triangle-free. 

It remains to observe that each $w \in N(U_0(U)) \cap U$ belongs to
some $S \in {\bf Loaded}(U)$ (as has been observed above). 
So, it can only be adjacent to $first(S)$.
\end{proof}

\begin{proof}
{\bf of Lemma \ref{lem:largeindep1}.}
First let us observe that
\begin{equation} \label{eq:mh1}
3^h < m(h) < 3^{h+1}
\end{equation}

Next we observe that
\begin{equation}\label{eq:mh2}
m(\lceil \alpha \cdot h\rceil+1)>m(h)^{\alpha}
\end{equation}

Indeed, by \eqref{eq:mh1}
$m(\lceil \alpha \cdot h \rceil+1)> 3^{\lceil \alpha \cdot h \rceil+1}> 3^{\alpha \cdot h+\alpha}=
(3^{h+1})^{\alpha}>m(h)^{\alpha}$. 

Set $h_0=\lceil \alpha \cdot h \rceil+2$. 
Then it is immediate from \eqref{eq:mh2} that
\begin{equation} \label{eq:mh3}
m(h_0-1)>m(h)^{\alpha}
\end{equation}

The next step is setting $h_1$.
Let us set $h_1=h_0+2 \lceil \log_3 k \rceil + \lceil \log  h \rceil $. 
We need $h$ to be large enough so that $h-h_1=\Omega((1-\alpha) \cdot h)$. 
With our settings, $h-h_1=h-\lceil \alpha h \rceil-
   (\lceil log h \rceil  + 2\lceil log_3 k \rceil+ 2)$
So, we can set $h$ to be large enough so that
$\lceil \log h \rceil +2\lceil \log_3 k \rceil+2 \leq (1-\alpha)/2 \cdot h$ or, to put it differently, 
$h \geq 2(\lceil \log h \rceil + 2\lceil \log_3 k \rceil +2)/(1-\alpha)$.

Now, we assume first that $\pi$ is $(h_0,h_1)$-\textsc{td}. 
According to Theorem \ref{th:maintopdowncomb},
there is a prefix $\pi_0$ such that
there is an independent ${\bf P} \subseteq {\bf P}(\pi_0,h_0)$
of size $1/260 \cdot (h-h_1) \cdot \sqrt{k}$. 
With our setting of $h_1$, we lower bound the size of ${\bf P}$ as 
$\Omega((1-\alpha) \cdot h \cdot \sqrt{k})$. 
It follows from Lemma \ref{lem:tripletopdown} that
$(\pi_0,U_0({\bf P}),U_1({\bf P}))$ as in the statement fo the lemma
is a target triple. Note that, by constructionm the rank of the triple
is $|{\bf P}|$. Hence the present lemma holds in the considered case, 

It remains to assume that $\pi$ is not $(h_0,h_1)$-\textsc{td}.  
By Lemma \ref{lem:winwin}, $\pi$ is $(h_0,h_2)$-\textsc{bu} where
$h_2 \geq h_0+ \lceil \log_3 k \rceil+\lceil \log h \rceil$. 
This means that there is $t \in V_{h_2}(h)$ and a prefix $\pi_0(t)$ of 
$\pi$ such that for each $S \in \mathcal{S}_{h_0}(h,k,t)$, $\pi_0(t) \cap S \neq \emptyset$
whereas $\pi_0(t) \cap V_{>h_0}(h,k,t)=\emptyset$. 
Clearly, the first property implies that $\pi_0(t)$ has a non-empty intersection with 
each $S \in {\bf S}_{h_0+1}(h,k,t)$. 
Let $\pi'_0$ be the shortest prefix of $\pi_0(t)$ subjecto the last property
and let $\pi_0$ be $\pi'_0$ without the last element. 
We conclude that $\pi_0$ is easy on $t$. 

We next apply Lemma \ref{lem:triplebottomup} 
to conclude existence of a target triple $(\pi_0,U_0(\pi_0),U_1(\pi_0))$. 
The rank of the triple is $|{\bf S}_{h_0+1}(h,k,root(\pi_0))|-1$.
We observe that $height(root(\pi_0))=h_2-1$. 
Hence $|{\bf S}_{h_0+1}(h,k,root(\pi_0))|-1=3^{h_2-1-h_0-1}-1=
3^{\lceil \log h \rceil+\lceil \log_3 k \rceil-2}-1=\Omega(k \cdot h)$
way above the lower bound required by the lemma. 
\end{proof}


\section{Conclusion} \label{sec:conclusion}
In this paper, we proved an \textsc{xp} lower bound for the size
of $\wedge_{d,\alpha}$-\textsc{fbbd}. Let us discuss why
the result does not generalize to \textsc{dnnf}s with imbalanced 
gates. To do this, let us recall the terminology in
the Preliminaries section. In particular, let $B$ 
be a $\wedge_{d,\alpha}$-\textsc{fbdd} representing $\varphi(T[h,k])$. 
Let $P$ be a target path of $B$ let ${\bf a}={\bf a}(P)$ and
let $u=u(P)$. Let ${\bf b}$ be an assignment over $Var(B_u)$
such that ${\bf a} \cup {\bf b}$ can be extended to a satisfying assignment of $f(B)$. 
A cornerstone property of $\wedge_d$-\textsc{fbdd} is that ${\bf b} \in \mathcal{S}(B_u)$. 
In particular, this property enables the use of Lemma \ref{lem:nobreak1} for connected subsets
of $Var(B_u)$ that are not fixed by ${\bf a}$ thus forbidding a decomposable conjunction gate to separate 
such a subset. 

The above property is lost for \textsc{dnnf}s where the underlying branching program is $1$-\textsc{nbp}
rather than \textsc{fbdd}. This is because the non-determinism allows the same assignment to be
carried along many different paths of the branching program and hence, if $B$ is a \textsc{dnnf}, $B_u$ does not need 
to have all ${\bf b}$ as above as satisfying assignments. As a result Lemma \ref{lem:nobreak1} does not
apply, unfixed connected sets may split and hence the whole machinery is not working. 

Yet, we conjecture that the non-determinism does not help to avoid 
the \textsc{xp} lower bound and a corresponding upgrade to the approach 
of this paper is possible. The upgrade may use ideas from \cite{AmarilliTCS,Bova14,Korhonen21}. 
The key intuition is similar to Lemma \ref{lem:elpos}. 
In particular, if a decomposable conjunction gate $u$ splits an unfixed
connected set then there is a set cover $C$ of the split edges such that every assignment carried 
through $u$ assigns $1$ to each vertex of $C$. This significantly reduces the probability 
that an assignment is carried through $u$. We thus believe tha  this line of reasoning will
lead to an upgrade of the bottleneck plus the union bound argument used in this paper. 


The main combinatorial statement Lemma \ref{lem:largeindep1},
in fact proves that, for sufficiently large $h$ and $k$,
the graph $T[h,k]$ has a large pathwidth with a `big' 
bag at a `bad' location. Indeed, if we consider a tree
decomposition of a graph $G$ in the form of an elimination forest
then, for a path decomposition, the elimination forest is
nothing else than a permutation $\pi$ of vertices of $G$. 
This permutation serves as the underlying path for the considered
decomposition. For each $u \in V(G)$, the bag of $u$, denoted by $bag(u)$,
is a set consisting of $u$ and all the proper predecessors of $u$ that
are adjacent to a successor of $u$. The component of $u$, denoted by
$cmp(u)$, is the suffix of $\pi$ starting at $u$. 
Lemma \ref{lem:largeindep1} effectively says that for any path decomposition
of $T[h,k]$, there is a vertex $u$ with $|bag(u)| \geq \Omega(\log n \sqrt{k})$
and such that the subgraph of $T[h,k]$ induced by $cmp(u)$ has each
connected component of size greater than $n^{\alpha}$. 
Thus Lemma \ref{lem:largeindep1} establishes an $\Omega(\log n \sqrt{k})$ lower
bound of a \emph{non-local} width parameter of $T[h,k]$. 
The non-locality means that a large value of the parameter requires not just presence 
of a large bag but also a `bad' location of such a bag in each path decomposition. 
Let us turn the idea around and consider what a small value of such a parameter
might mean for a tree or a path decomposition of a graph. A natural answer is that
existence of a decomposition where big bags may be present but only at `good' 
locations. 

We believe the above idea is worth exploration from the algorithmic perspective. 
In particular, such parameters may lead to parameterized algorithms 
with runtime $f(k) \cdot 2^{n^{\alpha}}$ where $0<\alpha<1$. 
This approach may be seen as a new concept of bidimensionality \cite{bidimensionality}
with $k$ and $\alpha$ being the input dimensions. 


\begin{thebibliography}{10}

\bibitem{AmarilliTCS}
Antoine Amarilli, Florent Capelli, Mika{\"{e}}l Monet, and Pierre Senellart.
\newblock Connecting knowledge compilation classes and width parameters.
\newblock {\em Theory Comput. Syst.}, 64(5):861--914, 2020.

\bibitem{BeameDNNF}
Paul Beame, Jerry Li, Sudeepa Roy, and Dan Suciu.
\newblock Lower bounds for exact model counting and applications in
  probabilistic databases.
\newblock In {\em Proceedings of the Twenty-Ninth Conference on Uncertainty in
  Artificial Intelligence, Bellevue, WA, USA, August 11-15, 2013}, 2013.

\bibitem{Bova14}
Simone Bova, Florent Capelli, Stefan Mengel, and Friedrich Slivovsky.
\newblock Expander cnfs have exponential {DNNF} size.
\newblock {\em CoRR}, abs/1411.1995, 2014.

\bibitem{DarwicheJACM}
Adnan Darwiche.
\newblock Decomposable negation normal form.
\newblock {\em J. ACM}, 48(4):608--647, 2001.

\bibitem{SDD}
Adnan Darwiche.
\newblock {S}{D}{D}: A new canonical representation of propositional knowledge
  bases.
\newblock In {\em 22nd International Joint Conference on Artificial
  Intelligence ({IJCAI})}, pages 819--826, 2011.
	
\bibitem{bidimensionality}
Erik D. Demaine, Fedor V. Fomin, Mohammadtaghi Hajiaghayi, Dimitrios M. Thilikos. 
\newblock Subexponential parameterized algorithms on bounded-genus graphs and $H$-minor-free graphs.
\newblock {\em J. ACM}, 52(6):866--893, 2005.

\bibitem{Diestel3}
Reinhard Diestel.
\newblock {\em Graph Theory, 3d Edition}, volume 173 of {\em Graduate texts in.
  mathematics}.
\newblock Springer, 2005.

\bibitem{HTWpaper}
Georg Gottlob, Nicola Leone, and Francesco Scarcello.
\newblock Hypertree decompositions and tractable queries.
\newblock {\em J. Comput. Syst. Sci.}, 64(3):579--627, 2002.

\bibitem{2009gottlob}
Georg Gottlob, Zolt{\'{a}}n Mikl{\'{o}}s, and Thomas Schwentick.
\newblock Generalized hypertree decompositions: {NP}-hardness and tractable
  variants.
\newblock {\em J. {ACM}}, 56(6):30:1--30:32, 2009.

\bibitem{Korhonen21}
Tuukka Korhonen.
\newblock Lower bounds on dynamic programming for maximum weight independent
  set.
\newblock In {\em {ICALP}2021}, volume 198 of {\em LIPIcs}, pages 87:1--87:14.
  Schloss Dagstuhl - Leibniz-Zentrum f{\"{u}}r Informatik, 2021.

\bibitem{ANDOBDDuneven}
Yong Lai, Dayou Liu, and Minghao Yin.
\newblock New canonical representations by augmenting obdds with conjunctive
  decomposition.
\newblock {\em J. Artif. Intell. Res.}, 58:453--521, 2017.

\bibitem{DesDNNF}
Umut Oztok and Adnan Darwiche.
\newblock On compiling {CNF} into decision-dnnf.
\newblock In {\em Principles and Practice of Constraint Programming - 20th
  International Conference, ({CP})}, pages 42--57, 2014.

\bibitem{RazgonAlgo}
Igor Razgon.
\newblock On the read-once property of branching programs and cnfs of bounded
  treewidth.
\newblock {\em Algorithmica}, 75(2):277--294, 2016.

\bibitem{obddtcompsys}
Igor Razgon.
\newblock On oblivious branching programs with bounded repetition that cannot
  efficiently compute cnfs of bounded treewidth.
\newblock {\em Theory Comput. Syst.}, 61(3):755--776, 2017.

\bibitem{VaThesis}
Martin Vatschelle.
\newblock {\em New width parameters of graphs}.
\newblock PhD thesis, Department of Informatics, University of Bergen, 2012.

\end{thebibliography}

\appendix
\section{Proof of Theorem \ref{th:decomp1}.}
Let $u$ be the source of $B$. Let $u_0$ and $u_1$
be the children of $u$. We assume w.l.o.g. that $(u,u_0)$
is an  edge of $P$. 
In case $u$ is labelled by a variable $x$, we assume 
w.l.o.g that $(u,u_0)$ is labelled with $0$.

We proceed by induction on the length of $P$.
Assume first that the length is $1$. 
Further on, assume that $u$ is labelled with a variable $x$.

Let ${\bf a} \in \mathcal{S}(B)|_{\{(x,0)\}}$. 
By definition of $\mathcal{S}(B)$, 
${\bf a} \in \mathcal{S}(B_{u_0}) \times \{0,1\}^{V_0}$
where $V_0=Var(B_{u_1}) \setminus Var(B_{u_0})=(Var(B) \setminus \{x\}) \setminus Var(B_{u_0})=V_0(P)$,
the last equality holds since $Alt(P)=\emptyset$. 
The converse also hods by definition of $\wedge_d$-\textsc{fbdd}. 
Hence, the theorem holds in the considered case. 

If $u$ is a $\wedge_d$-node then ${\bf a}(P)=\emptyset$. 
Hence $\mathcal{S}(B)_{{\bf a}(P)}=\mathcal{S}(B)$ which by definition is
$\mathcal{S}(B_{u_0}) \times \mathcal{S}(B_{u_1})$. 
As $Alt(P)=\{u_1\}$, the statement of the theorem follows for 
the considered case. 

Now, we consider the situation where $|P|>1$.
We let $P^-$ be the suffix of $P$ starting from $u_0$. 
Observe that $P^-$ is a target path of $B_{u_0}$ which we denote by $B^-$ for the sake of brevity. 
We let $V_0(P^-)$ and $Alt(P^-)$ be computed with respect to $B^-$
and we note that $u(P^-)=u(P)$. For the sake of simplicity, we consider only the case $Alt(P^-) \neq \emptyset$, 
the other case is similar. 
Then by the induction assumption

\begin{equation} \label{eq:decomp1}
\mathcal{S}(B^-)|_{{\bf a}(P^-)}=\mathcal{S}(B_{u(P)}) \times \prod_{v \in Alt(P^-)} \mathcal{S}(B_v) \times \{0,1\}^{V_0(P^-))}
\end{equation}

We also note that if $\mathcal{S}, \mathcal{S'}, \mathcal{S''}$ are sets
of assignments such that $\mathcal{S}=\mathcal{S'} \times \mathcal{S''}$,
${\bf g}$ is an assignment with $Var({\bf g}) \subseteq Var({\bf S'})$ then

\begin{equation} \label{eq:decomp2}
\mathcal{S}|_{\bf g}=\mathcal{S'}|_{\bf g} \times \mathcal{S''}
\end{equation}

Assume that $u$ is labelled by variable $x$. 
For the sake of brevity, we denote $\{(x,0)\}$ by ${\bf a}_0$
and recall from the induction base that
$\mathcal{S}|_{{\bf a}_0}=\mathcal{S}(B^-) \times \{0,1\}^{V_0}$.
Therefore
$\mathcal{S}(B)|_{{\bf a}(P)}=(\mathcal{S}|_{{\bf a}_0})|_{{\bf a}(P^-)}=
\mathcal{S}(B^-)|_{{\bf a}(P^-)} \times \{0,1\}^{V_0}$, the last equality 
follows from \eqref{eq:decomp2}. 
The theorem now follows by substitution of the right-had part of \eqref{eq:decomp1} 
instead of $\mathcal{S}(B^-)|_{{\bf a}(P^-)}$ and taking into account that
$Alt(P)=Alt(P^-)$ and $V_0(P)=V_0(P^-) \cup V_0$.

It remains to assume that $u$ is a $\wedge_d$-node. 
We note that ${\bf a}(P^-)={\bf a}(P)$. 
Therefore, by definition of $\wedge_d$-\textsc{fbdd} and 
\eqref{eq:decomp2}, 

\begin{equation} \label{eq:decomp3}
\mathcal{S}|_{{\bf a}(P)}=\mathcal{S}(B^-)_{{\bf a}(P)} \times \mathcal{S}(B_{u_1})
\end{equation}

We now substitute the right-hand side of \eqref{eq:decomp1} instead of 
$\mathcal{S}(B^-)_{{\bf a}(P)}$, take into account that, in the considered case, $V_0(P)=V_0(P^-)$,
and $Alt(P)=Alt(P^-) \cup \{u_1\}$, and observe that we have just obtained the statement of the 
theorem.



\end{document}